\documentclass[english,11pt]{article}
\usepackage[T1]{fontenc}
\usepackage[latin9]{inputenc}
\usepackage{geometry}
\usepackage{babel}
\usepackage{amsthm}
\usepackage{amsmath}
\usepackage{amssymb}
\usepackage[unicode= true,pdfusetitle,
 bookmarks= true,bookmarksnumbered= false,bookmarksopen= false,
 breaklinks= false,pdfborder= {0 0 1},backref= false,colorlinks= false]
 {hyperref}
\usepackage{breakurl, mathrsfs, appendix}
\usepackage{fancyhdr}
\usepackage{fullpage}
\usepackage{graphicx, subcaption}
\usepackage{color}

\usepackage[firstinits=false, backend=bibtex, maxbibnames=99]{biblatex}
\makeatletter

\newlength{\widebarargwidth}
\newlength{\widebarargheight}
\newlength{\widebarargdepth}
\DeclareRobustCommand{\widebar}[1]{%
	\settowidth{\widebarargwidth}{\ensuremath{#1}}%
	\settoheight{\widebarargheight}{\ensuremath{#1}}%
	\settodepth{\widebarargdepth}{\ensuremath{#1}}%
	\addtolength{\widebarargwidth}{-0.3\widebarargheight}%
	\addtolength{\widebarargwidth}{-0.3\widebarargdepth}%
	\makebox[0pt][l]{\hspace{0.3\widebarargheight}%
		\hspace{0.3\widebarargdepth}%
		\addtolength{\widebarargheight}{0.3ex}%
		\rule[\widebarargheight]{0.95\widebarargwidth}{0.1ex}}%
	{#1}}

\newcommand{\order}{\ensuremath{\mathcal{O}}}
\newcommand{\real}{\ensuremath{\mathbb{R}}}

\newcommand{\argmin}{\arg\!\min}

\newcommand{\inprod}[2]{\ensuremath{\langle #1 , \, #2 \rangle}}
\newcommand{\trinprod}[2]{\ensuremath{\langle \!\langle {#1}, \; {#2}
		\rangle \!\rangle}}

\newcommand{\row}[2]{\ensuremath{#1_{(#2,\cdot)}} }
\newcommand{\column}[2]{\ensuremath{#1_{(\cdot,#2)}} }
\newcommand{\entry}[3]{\ensuremath{#1_{(#2,#3)}}}

\newcommand{\matsnorm}[2]{|\!|\!| #1 | \! | \!|_{{#2}}}
\newcommand{\vecnorm}[2]{\| #1\|_{#2}}
\newcommand{\abs}[1]{| #1 |}
\newcommand{\twonorm}[1]{\vecnorm{#1}{2}}
\newcommand{\zeronorm}[1]{\vecnorm{#1}{0}}
\newcommand{\onenorm}[1]{\vecnorm{#1}{1}}
\newcommand{\opnorm}[1]{\ensuremath{\matsnorm{#1}{\mbox{\tiny{op}}}}}
\newcommand{\nucnorm}[1]{\ensuremath{\matsnorm{#1}{\mbox{\tiny{nuc}}}}}
\newcommand{\fronorm}[1]{\ensuremath{\matsnorm{#1}{\mbox{\tiny{F}}}}}
\newcommand{\infnorm}[1]{\ensuremath{\vecnorm{#1}{\infty}}}
\newcommand{\twoinfnorm}[1]{\ensuremath{\matsnorm{#1}{2,\infty}}}

\newcommand{\Loss}{\ensuremath{\mathcal{L}}}

\newcommand{\PlM}{\ensuremath{M}}
\newcommand{\PIU}{\ensuremath{U}}
\newcommand{\PIV}{\ensuremath{V}}
\newcommand{\PIS}{\ensuremath{S}}

\newcommand{\PIL}{\ensuremath{L}}
\newcommand{\PIR}{\ensuremath{R}}

\newcommand{\USet}{\ensuremath{\mathcal{U}}}
\newcommand{\VSet}{\ensuremath{\mathcal{V}}}

\newcommand{\Proj}{\ensuremath{\Pi}}

\newcommand{\SStar}{\ensuremath{ S^{*} }}
\newcommand{\MStar}{\ensuremath{ M^{*} }}
\newcommand{\InitS}{\ensuremath{S_{\text{init}}}}
\newcommand{\OptU}{\ensuremath{U^*}}
\newcommand{\OptV}{\ensuremath{V^*}}

\newcommand{\OptVT}{\ensuremath{V^{*\top}}}

\newcommand{\UniqueU}{\ensuremath{U_{\pi^*}}}
\newcommand{\UniqueV}{\ensuremath{V_{\pi^*}}}
\newcommand{\DeltaU}{\ensuremath{\Delta_{U}}}
\newcommand{\DeltaV}{\ensuremath{\Delta_{V}}}
\newcommand{\RotateM}{\ensuremath{Q}}
\newcommand{\RotateSet}[1]{\ensuremath{\mathbb{Q}_{#1}}}
\newcommand{\OptSet}{\ensuremath{\mathcal{E}}}
\newcommand{\ObserveSet}{\Phi}
\newcommand{\ObserveSparse}{\Omega^*_o}

\newcommand{\Trace}{\text{Tr}}

\newcommand{\ObserveMat}{\ensuremath{Y}}

\newcommand{\SVMat}{\ensuremath{\Sigma}}

\newcommand{\Sparset}[1]{\ensuremath{\mathcal{S}}_{#1}}
\newcommand{\Neighbor}[1]{\ensuremath{\mathbb{B}_{2}\left(#1\right)}}
\newcommand{\frobnorm}[1]{\ensuremath{\fronorm{#1}}}

\newcommand{\UnifSparse}[2]{\ensuremath{\mathcal{T}_{#2}\left[ #1\right]}}

\usepackage{algorithmicx}
\usepackage{algorithm}
\usepackage{algpseudocode}
\algnewcommand\algorithmicinput{\textbf{INPUT:}}
\algnewcommand\INPUT{\item[\algorithmicinput]}
\algnewcommand\algorithmicoutput{\textbf{OUTPUT:}}
\algnewcommand\OUTPUT{\item[\algorithmicoutput]}
\algnewcommand\algorithmicremark{\textsl{//}}
\algnewcommand\REMARK{\item[\algorithmicremark]}

\theoremstyle{plain}
\newtheorem{theorem}{\protect\theoremname}
\theoremstyle{remark}
\newtheorem{remark}{\protect\remarkname}
\theoremstyle{plain}
\newtheorem{lemma}{\protect\lemmaname}
\theoremstyle{plain}

\theoremstyle{plain}
\newtheorem{corollary}{\protect\corolaryname}

\providecommand{\lemmaname}{Lemma}
\providecommand{\remarkname}{Remark}
\providecommand{\theoremname}{Theorem}
\providecommand{\propositionname}{Proposition}
\providecommand{\corolaryname}{Corollary}

\makeatletter
\long\def\@makecaption#1#2{
        \vskip 0.8ex
        \setbox\@tempboxa\hbox{\small {\bf #1:} #2}
        \parindent 1.5em  %% How can we use the global value of this???
        \dimen0=\hsize
	       \advance\dimen0 by -3em
        \ifdim \wd\@tempboxa >\dimen0
                \hbox to \hsize{
                        \parindent 0em
                        \hfil 
                        \parbox{\dimen0}{\def\baselinestretch{0.96}\small
                                {\bf #1.} #2
                                } 
                        \hfil}
        \else \hbox to \hsize{\hfil \box\@tempboxa \hfil}
        \fi
        }
\makeatother

\begin{document}

\begin{center} {\LARGE{\bf{Fast Algorithms for Robust PCA via Gradient Descent
}}} \\

\vspace*{.3in}

{\large{
\begin{tabular}{ccccccc}
Xinyang Yi$^\ast$ && Dohyung Park$^\ast$ && Yudong Chen$^{\dagger}$ && Constantine Caramanis$^\ast$ 
\end{tabular}

\vspace*{.1in}

\begin{tabular}{ccc}
The University of Texas at Austin$^{\ast}$
& &  Cornell University$^{\dagger}$
\end{tabular}

\vspace*{.2in}

\begin{tabular}{c}
{\texttt{$\{$yixy,dhpark,constantine$\}$@utexas.edu$^{\ast}$ $\quad$ yudong.chen@cornell.edu$^{\dagger}$}}
\end{tabular} 
}}

\vspace*{.2in}

%\today

%\maketitle
\vspace*{.2in}
\begin{abstract}
We consider the problem of Robust PCA  in the fully and partially observed settings. Without corruptions, this is the well-known matrix completion problem. From a statistical standpoint this problem has been recently well-studied, and conditions on when recovery is possible (how many observations do we need, how many corruptions can we tolerate) via polynomial-time algorithms is by now understood. This paper presents and analyzes a non-convex optimization approach that greatly reduces the computational complexity of the above problems, compared to the best available algorithms. In particular, in the fully observed case, with $r$ denoting rank and $d$ dimension, we reduce the complexity from $\order(r^2d^2\log(1/\varepsilon))$ to $\order(rd^2\log(1/\varepsilon))$ -- a big savings when the rank is big. For the partially observed case, we show the complexity of our algorithm is no more than $\order(r^4d \log d \log(1/\varepsilon))$. Not only is this the best-known run-time for a provable algorithm under partial observation, but in the setting where $r$ is small compared to $d$, it also allows for near-linear-in-$d$ run-time that can be exploited in the fully-observed case as well, by simply running our algorithm on a subset of the observations. 
\end{abstract}
\end{center}

%%%%%%%%%%%%%%%%%%%%%%%%%%%%%%%%%%%%%%%%%%%%%%%%%%%%%%%%%%%%%%%%%%%%%%%%%%%

\section{Introduction}
{\em Principal component analysis} (PCA) aims to find a low rank subspace that best-approximates a data matrix $Y \in \real^{d_1 \times d_2}$. The simple and standard method of PCA by {\em singular value decomposition} (SVD) fails in many modern data problems due to missing and corrupted entries, as well as sheer scale of the problem. Indeed, SVD is highly sensitive to outliers by virtue of the squared-error criterion it minimizes. Moreover, its running time scales as $\order(r d^2)$ to recover a rank $r$ approximation of a $d$-by-$d$ matrix. 

While there have been recent results developing provably robust algorithms for PCA (e.g., \cite{candes2011robust,Xu2012RPCA}), the running times range from $\order(r^2 d^2)$ to $\order(d^3)$\footnote{For precise dependence on error and other factors, please see details below.} and hence are significantly worse than SVD. Meanwhile, the literature developing sub-quadratic algorithms for PCA (e.g., \cite{frieze2004fast, clarkson2013low,bhojanapalli2015tighter}) seems unable to guarantee robustness to outliers or missing data.

Our contribution lies precisely in this area: provably robust algorithms for PCA with improved run-time. Specifically, we provide an efficient algorithm with running time that matches SVD while nearly matching the best-known robustness guarantees. In the case where rank is small compared to dimension, we develop an algorithm with running time that is nearly linear in the dimension. This last algorithm works by subsampling the data, and therefore we also show that our algorithm solves the Robust PCA problem with partial observations (a generalization of matrix completion and Robust PCA). 

\subsection{The Model and Related Work} 
We consider the following setting for robust PCA. Suppose we are given a matrix $Y \in \real^{d_1 \times d_2}$ that has decomposition $Y = \MStar + \SStar$, where $\MStar$ is  a rank $r$ matrix and $\SStar$ is a sparse corruption matrix containing entries with arbitrary magnitude. The goal is to recover $\MStar$ and $\SStar$ from $Y$. To ease notation, we let $d_1 = d_2 = d$ in the remainder of this section. 

Provable solutions for this model are first provided in the works of \cite{chandrasekaran2011rank} and \cite{candes2011robust}. They propose to solve this problem by {\em convex relaxation}:
\begin{equation} \label{eq:convex_rpca}
\min_{M, S} \nucnorm{M} + \lambda \|S\|_1, ~\text{s.t.}~ Y = M + S,
\end{equation}
where $\nucnorm{M}$ denotes the nuclear norm of $M$. Despite analyzing the same method, the corruption models in \cite{candes2011robust} and \cite{chandrasekaran2011rank} differ. In \cite{candes2011robust}, the authors consider the setting where the entries of $\MStar$ are corrupted at random with probability $\alpha$. They show their method succeeds in exact recovery with $\alpha$ as large as $0.1$, which indicates they can tolerate a constant fraction of corruptions. Work in \cite{chandrasekaran2011rank} considers a {\em deterministic corruption model}, where nonzero entries of $\SStar$ can have arbitrary position, but the sparsity of each row and column does not exceed $\alpha d$. They prove that for exact recovery, it can allow $\alpha = \order(1/(\mu r\sqrt{d}))$. This was subsequently further improved to $\alpha = \order(1/(\mu r))$, which is in fact optimal  \cite{chen2011LSarxiv, hsu2011robust}. Here, $\mu$ represents the incoherence of $\MStar$ (see Section \ref{sec:problem_setup} for details). In this paper, we follow this latter line and focus on the deterministic corruption model.

The state-of-the-art solver \cite{LinChenMa10} for \eqref{eq:convex_rpca} has time complexity $\order(d^3/\varepsilon)$ to achieve error $\varepsilon$, and is thus much slower than SVD, and prohibitive for even modest values of $d$. Work in \cite{netrapalli2014non} considers the deterministic corruption model, and improves this running time without sacrificing the robustness guarantee on $\alpha$. They propose an {\em alternating projection} (AltProj) method to estimate the low rank and sparse structures iteratively and simultaneously, and show their algorithm  has complexity $\order(r^2d^2\log(1/\varepsilon))$, which is faster than the convex approach but still slower than SVD.

Non-convex approaches have recently seen numerous developments for applications in low-rank estimation, including alternating minimization (see e.g.\ \cite{jain2013low,hardt2014understanding,gu2016low}) and gradient descent  (see e.g.\ \cite{bhojanapalli2015dropping,chen2015fast, sun2015guaranteed,tu2015low,zhao2015nonconvex,zheng2015convergent}). These works have fast running times, yet do not provide robustness guarantees. One exception is \cite{chen2015fast}, where the authors analyze a row-wise $\ell_1$ projection method for recovering $\SStar$. Their analysis hinges on positive semidefinite $\MStar$, and the algorithm requires prior knowledge of the $\ell_1$ norm of every row of $\SStar$ and is thus prohibitive in practice. Another exception is work \cite{gu2016low}, which analyzes alternating minimization plus an overall sparse projection. Their algorithm is shown to tolerate at most a fraction of $\alpha = \order{(1/(\mu^{2/3} r^{2/3} d))}$ corruptions. As we discuss in Section \ref{sec:contribution}, we can allow $\SStar$ to have much higher sparsity $\alpha = \order{(1/(\mu r^{1.5}))}$, which is close to optimal. 

After the initial post of our paper on arxiv, Cherapanamjeri et~al.\ \cite{cherapanamjeri2016nearly} posted their paper on solving robust PCA with partial observations by a method modified from AltProj. Their approach is shown to have optimal robustness. The sample complexity they established is $\order(\mu^2r^2d\log^2 d\log^2(1/\varepsilon))$, which depends on the estimation error $\varepsilon$, since the analysis of AltProj under partial observations requires sampling splitting. In contrast, the sample complexity $\order(\mu^2r^2d\log d)$ of our approach (see Corollary \ref{cor:rpca_partial}) does not have such dependence. Moreover, their method requires computing rank-$r$ SVD of the observed matrices in every round of iterations. Our algorithm only needs to compute rank-$r$ SVD (approximately) once in the initialization step.

It is worth mentioning other works that obtain provable guarantees of non-convex algorithms or problems including phase retrieval \cite{candes2015phase,chen2015solving,zhang2016provable}, EM algorithms \cite{balakrishnan2014statistical, wang2015high, yi2015regularized}, tensor decompositions \cite{anandkumar2014tensor} and second order method \cite{sun2015nonconvex}. It might be interesting to bring robust considerations to these works.

\subsection{Our Contributions} \label{sec:contribution}
In this paper, we develop efficient non-convex algorithms for robust PCA. We propose a novel algorithm based on the projected gradient method on the factorized space. We also extend it to solve robust PCA in the setting with partial observations, i.e., in addition to gross corruptions, the data matrix has a large number of missing values. Our main contributions are summarized as follows.\footnote{To ease presentation, the discussion here assumes $\MStar$ has constant condition number, whereas our results below show the dependence on condition number explicitly.}
\begin{enumerate}
\item We propose a novel sparse estimator for the setting of deterministic corruptions. For the low-rank structure to be identifiable, it is natural to assume that deterministic corruptions are ``spread out'' (no more than some number in each row/column). We leverage this information in a simple but critical algorithmic idea, that is tied to the ultimate complexity advantages our algorithm delivers.

\item Based on the proposed sparse estimator, we propose a projected gradient method on the matrix factorized space. While non-convex, the algorithm is shown to enjoy linear convergence under proper initialization. Along with a new initialization method, we show that robust PCA can be solved within complexity $\order(rd^2\log(1/\varepsilon))$ while ensuring robustness $\alpha = \order(1/(\mu r^{1.5}))$. Our algorithm is thus faster than the best previous known algorithm by a factor of $r$, and enjoys superior empirical performance as well.

\item Algorithms for Robust PCA with partial observations still rely on a computationally expensive convex approach, as apparently this problem has evaded treatment by non-convex methods. We consider precisely this problem. In a nutshell, we show that our gradient method succeeds (it is guaranteed to produce the subspace of $\MStar$) even when run on no more than $\order(\mu^2r^2d\log d)$ random entries of $Y$. The computational cost is $\order(\mu^3r^4d\log d \log(1/\varepsilon))$. When rank $r$ is small compared to the dimension $d$, in fact this dramatically improves on our bound above, as our cost becomes nearly linear in $d$. We show, moreover, that this savings and robustness to erasures comes at {\em no cost in the robustness guarantee} for the deterministic (gross) corruptions. While this demonstrates our algorithm is robust to both outliers and erasures, it also provides a way to reduce computational costs even in the fully observed setting, when $r$ is small.

\item An immediate corollary of the above result provides a guarantee for exact matrix completion, with general rectangular matrices, using $\order(\mu^2r^2d\log d)$ observed entries and $\order(\mu^3r^4d\log d \log(1/\varepsilon))$ time, thereby improving on existing results in~\cite{chen2015fast,sun2015guaranteed}.
\end{enumerate}

\subsection{Organization and Notation}
The remainder of this paper is organized as follows. In Section \ref{sec:problem_setup}, we formally describe our problem and assumptions. In Section \ref{sec:algorithm}, we present and describe our algorithms for fully (Algorithm \ref{alg:rpca}) and partially (Algorithm \ref{alg:rpca_partial}) observed settings. In Section \ref{sec:analyze_rpca}, we establish theoretical guarantees of Algorithm \ref{alg:rpca}. The theory for partially observed setting are presented in Section \ref{sec:rpca_partial}. The numerical results are collected in Section \ref{sec:num_results}. Sections \ref{sec:proofs}, \ref{sec:proof_lemmas} and Appendix \ref{sec:supporting_lemma} contain all the proofs and technical lemmas.

For any index set $\Omega \subseteq [d_1]\times [d_2]$, we let $\row{\Omega}{i} := \left\{(i, j) \in \Omega ~\big|~ j \in [d_2]\right\}$, $\column{\Omega}{j} := \left\{(i, j) \in \Omega ~\big|~ i \in [d_1]\right\}$. For any matrix $A \in \real^{d_1 \times d_2}$, we denote its projector onto support $\Omega$ by $\Proj_{\Omega}\left(A\right)$, i.e., the $(i,j)$-th entry of $\Proj_{\Omega}\left(A\right)$ is equal to $A$ if $(i,j) \in \Omega$ and zero otherwise. The $i$-th row and $j$-th column of $A$ are denoted by $\row{A}{i}$ and $\column{A}{j}$. The $(i,j)$-th entry is denoted as $\entry{A}{i}{j}$. Operator norm of $A$ is $\opnorm{A}$. Frobenius norm of $A$ is $\frobnorm{A}$. The $\ell_a/\ell_b$ norm of $A$ is denoted by $\matsnorm{A}{b,a}$, i.e., the $\ell_a$ norm of the vector formed by the $\ell_b$ norm of every row. For instance, $\twoinfnorm{A}$ stands for $\max_{i \in [d_1]}\twonorm{\row{A}{i}}$.

\section{Problem Setup} \label{sec:problem_setup}
We consider the problem where we observe a matrix $\ObserveMat \in \real^{d_1\times d_2}$ that satisfies $\ObserveMat = \MStar + \SStar$, where $\MStar$ has rank $r$, and $\SStar$ is corruption matrix with sparse support. Our goal is to recover $\MStar$ and $\SStar$. In the partially observed setting, in addition to sparse corruptions, we have erasures. We assume that each entry of $\MStar + \SStar$ is revealed independently with probability $p \in (0,1)$. In particular, for any $(i,j) \in [d_1]\times [d_2]$, we consider the Bernoulli model where
\begin{equation} \label{eq:Y}
Y_{(i,j)} = \begin{cases}
(\MStar + \SStar)_{(i,j)}, & ~ \text{with probability}~ p; \\
*,&  ~\text{otherwise}.
\end{cases}
\end{equation}
We denote the support of $\ObserveMat$ by $\ObserveSet = \{(i,j) ~|~ Y_{(i,j)} \neq *\}$. Note that we assume $\SStar$ is not adaptive to $\ObserveSet$. As is well-understood thanks to work in matrix completion, this task is impossible in general -- we need to guarantee that $\MStar$ is not both low-rank and sparse. To avoid such identifiability issues, we make the following standard assumptions on $\MStar$ and $\SStar$: (i) $\MStar$ is not near-sparse or ``spiky.'' We impose this by requiring $\MStar$ to be $\mu$-incoherent -- given a singular value decomposition (SVD)\footnote{Throughout this paper, we refer to SVD of rank $r$ matrix by form $\PIL \Sigma \PIR^{\top}$ where $\Sigma \in \real^{r \times r}$.}  $\MStar = \PIL^* \SVMat^* \PIR^{*\top}$, we assume that
\[
\twoinfnorm{\PIL^*} \leq \sqrt{\frac{\mu r}{d_1}}, ~~ \twoinfnorm{\PIR^*} \leq \sqrt{\frac{\mu r}{d_2}}.
\] 
(ii) The entries of $\SStar$ are ``spread out'' -- for $\alpha \in [0,1)$, we assume $\SStar \in \Sparset{\alpha}$, where
\begin{equation} \label{eq:S_alpha}
\Sparset{\alpha} := \left\{ A \in \real^{d_1 \times d_2} ~\big|~ \zeronorm{\row{A}{i}} \leq \alpha d_2~\text{for all}~ i \in [d_1]~; \zeronorm{\column{A}{j}} \leq \alpha d_1~\text{for all}~ j \in [d_2]  \right\}.
\end{equation}
In other words, $\SStar$ contains at most $\alpha$-fraction nonzero entries per row and column.

%%%%%%%%%%%%%%%%%%%%%%%%%%%%%%%%%%%%%%%%%%%%%%%%%%%%%%%%%%%%%%%%%%%%%%%%%%%

\section{Algorithms} \label{sec:algorithm}
For both the full and partial observation settings, our method proceeds in two phases. In the first phase, we use a new sorting-based sparse estimator to produce a rough estimate $\InitS$ for $S^*$ based on the observed matrix $\ObserveMat$, and then find a rank $r$ matrix factorized as $\PIU_{0}\PIV_{0}^{\top}$ that is a rough estimate of $\MStar$ by performing SVD on ($\ObserveMat - \InitS$). In the second phase, given $(\PIU_{0}, \PIV_{0})$, we perform an iterative method to produce series $\{(\PIU_t, \PIV_t)\}_{t=0}^{\infty}$. In each step $t$, we first apply our sparse estimator to produce a sparse matrix $\PIS_{t}$ based on $(\PIU_{t}, \PIV_{t})$, and then perform a projected gradient descent step on the low-rank factorized space to produce $(\PIU_{t+1}, \PIV_{t+1})$. This flow is the same for full and partial observations, though a few details differ. Algorithm \ref{alg:rpca} gives the full observation algorithm, and Algorithm \ref{alg:rpca_partial} gives the partial observation algorithm. We now describe the key details of each algorithm.

\paragraph{Sparse Estimation.} A natural idea is to keep those entries of residual matrix $\ObserveMat - \PlM$ that have large magnitude. At the same time, we need to make use of the dispersed property of $\Sparset{\alpha}$ that every column and row contain at most $\alpha$-fraction of nonzero entries. Motivated by these two principles, we introduce the following sparsification operator: For any matrix $A \in \real^{d_1 \times d_2}$: for all $(i,j) \in [d_1]\times [d_2]$, we let
\begin{equation} \label{eq:sparse_estimator}
\UnifSparse{A}{\alpha} := \begin{cases} \entry{A}{i}{j}, & \text{if}~ \abs{\entry{A}{i}{j}} \geq \abs{\row{A}{i}^{(\alpha d_2)}} ~\text{and}~ \abs{\entry{A}{i}{j}} \geq \abs{\column{A}{j}^{(\alpha d_1)}} \\
0, & \text{otherwise}\end{cases},
\end{equation}
where $\row{A}{i}^{(k)}$ and  $\column{A}{j}^{(k)}$ denote the elements of $\row{A}{i}$ and $\column{A}{j}$ that have the $k$-th largest magnitude respectively. In other words, we choose to keep those elements that are simultaneously among the largest $\alpha$-fraction entries in the corresponding row and column. In the case of entries having identical magnitude, we break ties arbitrarily. It is thus guaranteed that $\UnifSparse{A}{\alpha} \in \Sparset{\alpha}$.

\begin{algorithm}[H]
	\caption{Fast RPCA}
	\label{alg:rpca}
	\begin{algorithmic}[1]
		\INPUT Observed matrix $\ObserveMat$ with rank $r$ and corruption fraction $\alpha$; parameters $\gamma, \eta$; number of iterations $T$.
		\vskip .05in
		\REMARK  {\em Phase I: Initialization.}
		\State $\InitS \leftarrow \UnifSparse{Y}{\alpha}$ \hskip .3in  {\em// see \eqref{eq:sparse_estimator} for the definition of $\UnifSparse{\cdot}{\alpha}$.}
		\State $[\PIL, \SVMat, \PIR] \leftarrow \text{SVD}_{r}[Y - \InitS]$ \footnotemark
		\State $\PIU_0 \leftarrow \PIL\SVMat^{1/2}$, $\PIV_0 \leftarrow \PIR\SVMat^{1/2}$. Let $\USet, \VSet$ be defined according to \eqref{eq:set}.
		\vskip .05in
		\REMARK {\em Phase II: Gradient based iterations.}
		\State $\PIU_0 \leftarrow \Proj_{\USet}\left(\PIU_0\right)$, $\PIV_0 \leftarrow \Proj_{\VSet}\left(\PIV_0\right)$
		\For{$t = 0,1,\ldots, T-1$}
		\State $\PIS_{t} \leftarrow \UnifSparse{Y - \PIU_{t}\PIV_{t}^{\top}}{\gamma \alpha}$  
		\State $\PIU_{t+1} \leftarrow \Proj_{\USet}\left(\PIU_{t} - \eta \nabla_{\PIU}\Loss(\PIU_{t}, \PIV_{t}; \PIS_{t}) - \frac{1}{2}\eta\PIU_{t}(\PIU_{t}^{\top}\PIU_{t} - \PIV_{t}^{\top}\PIV_{t})\right)$
		\State $\PIV_{t+1} \leftarrow \Proj_{\VSet}\left(\PIV_{t} - \eta \nabla_{\PIV}\Loss(\PIU_{t}, \PIV_{t}; \PIS_{t}) - \frac{1}{2}\eta\PIV_{t}(\PIV_{t}^{\top}\PIV_{t} - \PIU_{t}^{\top}\PIU_{t})\right)$
		\EndFor
		\OUTPUT $(\PIU_T, \PIV_T)$
	\end{algorithmic}
\end{algorithm} 
\footnotetext{$\text{SVD}_r[A]$ stands for computing a rank-$r$ SVD of matrix $A$, i.e., finding the top $r$ singular values and vectors of $A$. Note that we only need to compute rank-$r$ SVD approximately (see the initialization error requirement in Theorem \ref{thm:initialization}) so that we can leverage fast iterative approaches such as block power method and Krylov subspace methods \cite{saad2011numerical}.}

\paragraph{Initialization.} In the {\em fully observed} setting, we compute $\InitS$ based on $\ObserveMat$ as $\InitS = \UnifSparse{Y}{\alpha}$. In the {\em partially observed} setting with sampling rate $p$, we let $\InitS = \UnifSparse{Y}{2p\alpha}$. In both cases, we then set $\PIU_{0} = \PIL \SVMat^{1/2}$ and $\PIV_{0} = \PIR \SVMat^{1/2}$, where $\PIL \SVMat \PIR^{\top}$ is an SVD of the best rank $r$ approximation of $Y - \InitS$. 

\paragraph{Gradient Method on Factorized Space.} After initialization, we proceed by projected gradient descent. To do this, we define loss functions explicitly in the factored space, i.e., in terms of $\PIU, \PIV$ and $\PIS$:
\begin{eqnarray}
\label{eq:loss} 
\Loss(\PIU, \PIV; \PIS) &:=& \frac{1}{2}\frobnorm{\PIU\PIV^{\top} + \PIS - \ObserveMat}^2, \qquad \mbox{(fully observed)} \\
\label{eq:loss_partial}
\widetilde{\Loss}(\PIU, \PIV; \PIS) &:=& \frac{1}{2p}\frobnorm{\Proj_{\ObserveSet}\left(\PIU\PIV^{\top} + \PIS - \ObserveMat\right)}^2. \qquad \mbox{(partially observed)}
\end{eqnarray}
Recall that our goal is to recover $\MStar$ that satisfies the $\mu$-incoherent condition. Given an SVD $\MStar = \PIL^* \SVMat \PIR^{*\top}$, we expect that the solution $(\PIU, \PIV)$ is close to $(\PIL^*\SVMat^{1/2}, \PIR^*\SVMat^{1/2})$ up to some rotation. In order to serve such $\mu$-incoherent structure, it is natural to put constraints on the row norms of $\PIU, \PIV$ based on $\opnorm{\MStar}$. As $\opnorm{\MStar}$ is unavailable, given $\PIU_0, \PIV_0$ computed in the first phase, we rely on the sets $\USet$, $\VSet$ defined as
\begin{equation} \label{eq:set}
\USet := \left\{A \in \real^{d_1 \times r}~\big|~ \twoinfnorm{A} \leq \sqrt{\frac{2\mu r}{d_1}}\opnorm{\PIU_0}\right\}, ~
\VSet := \left\{A \in \real^{d_2 \times r}~\big|~ \twoinfnorm{A} \leq \sqrt{\frac{2\mu r}{d_2}}\opnorm{\PIV_0}\right\}.
\end{equation}
Now we consider the following optimization problems with constraints:
\begin{eqnarray} \label{eq:opt}
\min_{\PIU \in \USet, \PIV \in \VSet, \PIS \in \Sparset{\alpha}} && \Loss(\PIU, \PIV; \PIS) + \frac{1}{8}\frobnorm{\PIU^{\top}\PIU - \PIV^{\top}\PIV}^2, \qquad \mbox{(fully observed)} \\
\label{eq:opt_partial}
\min_{\PIU \in \USet, \PIV \in \VSet, \PIS \in \Sparset{p\alpha}} && \widetilde{\Loss}(\PIU, \PIV; \PIS) + \frac{1}{64}\frobnorm{\PIU^{\top}\PIU - \PIV^{\top}\PIV}^2. \qquad \mbox{(partially observed)}
\end{eqnarray}
The regularization term in the objectives above is used to encourage that $\PIU$ and $\PIV$ have the same scale. Given $(\PIU_0,\PIV_0)$, we propose the following iterative method to produce series $\{(\PIU_t,\PIV_t)\}_{t=0}^{\infty}$ and $\{\PIS_t\}_{t=0}^{\infty}$. We give the details for the fully observed case -- the partially observed case is similar. For $t = 0,1,\ldots$, we update
$ \PIS_{t} $ using the sparse estimator $\PIS_{t} = \UnifSparse{Y - \PIU_{t}\PIV_{t}^{\top}}{\gamma \alpha}$, followed by a projected gradient update on $ \PIU_{t} $ and $ \PIV_t $
\begin{align*}
\PIU_{t+1} & = \Proj_{\USet}\left(\PIU_{t} - \eta \nabla_{\PIU}\Loss(\PIU_{t}, \PIV_{t}; \PIS_{t}) - \frac{1}{2}\eta\PIU_{t}(\PIU_{t}^{\top}\PIU_{t} - \PIV_{t}^{\top}\PIV_{t})\right), \\
\PIV_{t+1} & = \Proj_{\VSet}\left(\PIV_{t} - \eta \nabla_{\PIV}\Loss(\PIU_{t}, \PIV_{t}; \PIS_{t}) - \frac{1}{2}\eta\PIV_{t}(\PIV_{t}^{\top}\PIV_{t} - \PIU_{t}^{\top}\PIU_{t})\right).
\end{align*}
Here $\alpha$ is the model parameter that characterizes the corruption fraction, $\gamma$ and $\eta$ are algorithmic tunning parameters, which we specify in our analysis.  Essentially, the above algorithm corresponds to applying projected gradient method to optimize \eqref{eq:opt}, where $S$ is replaced by the aforementioned sparse estimator in each step.

\begin{algorithm}[H]
	\caption{Fast RPCA with partial observations}
	\label{alg:rpca_partial}
	\begin{algorithmic}[1]
		\INPUT Observed matrix $\ObserveMat$ with support $\ObserveSet$; parameters $\tau, \gamma, \eta$; number of iterations $T$.
		\vskip .05in
		\REMARK  {\em Phase I: Initialization.}
		\State $\InitS \leftarrow \UnifSparse{\Proj_{\Phi}(Y)}{2p\alpha}$
		\State $[\PIL, \SVMat, \PIR] \leftarrow \text{SVD}_{r}[\frac{1}{p}(Y - \InitS)]$
		\State $\PIU_0 \leftarrow \PIL\SVMat^{1/2}$, $\PIV_0 \leftarrow \PIR\SVMat^{1/2}$. Let $\USet, \VSet$ be defined according to \eqref{eq:set}.
		\vskip .05in
		\REMARK {\em Phase II: Gradient based iterations.}
		\State $\PIU_0 \leftarrow \Proj_{\USet}\left(\PIU_0\right)$, $\PIV_0 \leftarrow \Proj_{\VSet}\left(\PIV_0\right)$
		\For{$t = 0,1,\ldots, T-1$}
		\State $\PIS_{t} \leftarrow \UnifSparse{\Proj_{\ObserveSet}\left(Y - \PIU_{t}\PIV_{t}^{\top}\right)}{\gamma p \alpha}$
		\State $\PIU_{t+1} \leftarrow \Proj_{\USet}\left(\PIU_{t} - \eta \nabla_{\PIU}\widetilde{\Loss}(\PIU_{t}, \PIV_{t}; \PIS_{t}) - \frac{1}{16}\eta\PIU_{t}(\PIU_{t}^{\top}\PIU_{t} - \PIV_{t}^{\top}\PIV_{t})\right)$
		\State $\PIV_{t+1} \leftarrow \Proj_{\VSet}\left(\PIV_{t} - \eta \nabla_{\PIV}\widetilde{\Loss}(\PIU_{t}, \PIV_{t}; \PIS_{t}) - \frac{1}{16}\eta\PIV_{t}(\PIV_{t}^{\top}\PIV_{t} - \PIU_{t}^{\top}\PIU_{t})\right)$
		\EndFor
		\OUTPUT $(\PIU_T, \PIV_T)$
	\end{algorithmic}
\end{algorithm} 

%%%%%%%%%%%%%%%%%%%%%%%%%%%%%%%%%%%%%%%%%%%%%%%%%%%%%%%%%%%%%%%%%%%%%%

\section{Main Results}

In this section, we establish theoretical guarantees for Algorithm~\ref{alg:rpca} in the fully observed setting and for Algorithm~\ref{alg:rpca_partial} in the partially observed setting.

\subsection{Analysis of Algorithm \ref{alg:rpca}} \label{sec:analyze_rpca}
We begin with some definitions and notation. It is important to define a proper error metric because the optimal solution corresponds to a manifold and there are many distinguished pairs $(\PIU, \PIV)$ that minimize \eqref{eq:opt}. Given the SVD  of the true low-rank matrix $\MStar = \PIL^*\SVMat^* \PIR^{*\top}$, we let $\OptU := \PIL^*\SVMat^{*1/2}$ and $\OptV := \PIR^*\SVMat^{*1/2}$.  We also let $\sigma_1^* \geq \sigma_2^* \geq \ldots \geq \sigma_r^*$ be sorted nonzero singular values of $\MStar$, and denote the condition number of $\MStar$ by $\kappa$, i.e., $\kappa := \sigma_1^*/\sigma_r^*$. We define estimation error $d(\PIU, \PIV; \OptU, \OptV)$ as the minimal Frobenius norm between $(\PIU, \PIV)$ and $(\OptU, \OptV)$ with respect to the optimal rotation, namely
\begin{equation} \label{eq:error_metric}
d(\PIU, \PIV; \OptU, \OptV) := \min_{\RotateM \in \RotateSet{r}} \sqrt{ \frobnorm{\PIU - \OptU \RotateM}^2 + \frobnorm{\PIV - \OptV \RotateM}^2},
\end{equation}
for $\RotateSet{r}$ the set of $r$-by-$r$ orthonormal matrices. This metric controls reconstruction error, as
\begin{equation} \label{eq:error_relation}
\frobnorm{\PIU\PIV^{\top} - \MStar} \lesssim \sqrt{\sigma_1^*}d(\PIU, \PIV; \OptU, \OptV),
\end{equation}
when $d(\PIU, \PIV; \OptU, \OptV) \leq \sqrt{\sigma_1^*}$. We denote the local region around the optimum $(\OptU, \OptV)$ with radius $\omega$ as
\[
\Neighbor{\omega} := \left\{(\PIU,\PIV) \in \real^{d_1\times r} \times \real^{d_2 \times r} ~\big|~ d(\PIU,\PIV;\OptU, \OptV) \leq \omega\right\}.
\]

The next two theorems provide guarantees for the initialization phase and gradient iterations, respectively, of Algorithm \ref{alg:rpca}. The proofs are given in Sections \ref{proof:thm:initialization} and \ref{proof:thm:convergence}.

\begin{theorem}[Initialization]\label{thm:initialization} Consider the paired $(\PIU_0,\PIV_0)$ produced in the first phase of Algorithm \ref{alg:rpca}. If $\alpha \leq 1/(16\kappa \mu r)$,
	we have 
	\[
	d(\PIU_0, \PIV_0; \OptU, \OptV) \leq 28 \sqrt{\kappa} \alpha \mu r\sqrt{r} \sqrt{\sigma_1^*}.
	\]
\end{theorem}

\begin{theorem}[Convergence] \label{thm:convergence}
	Consider the second phase of Algorithm \ref{alg:rpca}. Suppose we choose $\gamma = 2$ and $\eta = c/\sigma_1^*$ for any $c \leq 1/36$. There exist constants $c_1, c_2$ such that when $\alpha \leq c_1/(\kappa^2 \mu r)$,
	given any $(\PIU_0, \PIV_0) \in \Neighbor{c_2\sqrt{\sigma^*_r/\kappa}}$, the iterates $\{(\PIU_t, \PIV_t)\}_{t=0}^{\infty}$ satisfy
	\[
	d^2(U_t, V_t; \OptU, \OptV) \leq \left(1 - \frac{c}{8\kappa}\right)^td^2(U_0,V_0;\OptU, \OptV).
	\]
\end{theorem}
Therefore, using proper initialization and step size, the gradient iteration converges at a linear rate with a constant contraction factor $1 - \order(1/\kappa)$. To obtain relative precision $\varepsilon$ compared to the initial error, it suffices to perform $O(\kappa \log(1/\varepsilon))$ iterations. Note that the step size is chosen according to $1/\sigma_1^*$. When $\alpha \lesssim 1/(\mu \sqrt{\kappa r^3})$,  Theorem \ref{thm:initialization} and the inequality \eqref{eq:error_relation} together imply that $\opnorm{\PIU_0\PIV_0^{\top} - \MStar} \leq \frac{1}{2}\sigma_1^*$. Hence we can set the step size as $\eta = \order(1/\sigma_1(\PIU_0\PIV_0^{\top}))$ using being the top singular value $ \sigma_1(\PIU_0\PIV_0^{\top}) $ of the matrix $ \PIU_0 \PIV_0^\top $ 

Combining Theorems \ref{thm:initialization} and \ref{thm:convergence} implies the following result, proved in Section \ref{proof:cor:rpca}, that provides an overall guarantee for Algorithm \ref{alg:rpca}.

\begin{corollary} \label{cor:rpca}
	Suppose that
	\[
	\alpha \leq c\min\left\{\frac{1}{\mu\sqrt{\kappa r}^3}, \frac{1}{\mu \kappa^2r}\right\}
	\] 
	for some constant $c$. Then for any $\varepsilon \in (0,1)$, Algorithm \ref{alg:rpca} with $T = \order(\kappa \log (1/\varepsilon))$ outputs a pair $(\PIU_{T}, \PIV_{T})$ that satisfies
	\begin{equation} \label{eq:final_error}
	\frobnorm{\PIU_{T}\PIV_{T}^{\top} - \MStar} \leq \varepsilon\cdot\sigma_r^*.
	\end{equation}
\end{corollary}

\vskip .05in
\begin{remark}[Time Complexity] \label{remark:complexity} For simplicity we assume $d_1 = d_2 = d$. Our sparse estimator \eqref{eq:sparse_estimator} can be implemented by finding the top $\alpha d$ elements of each row and column via partial quick sort, which has running time $\order({d^2\log(\alpha d)})$. Performing rank-$r$ SVD in the first phase and computing the gradient in each iteration both have complexity $\order(rd^2)$.\footnote{In fact, it suffices to compute the best rank-$ r $ approximation with running time independent of the eigen gap.} Algorithm \ref{alg:rpca} thus has total running time $\order(\kappa rd^2\log(1/\varepsilon))$ for achieving an $ \epsilon $ accuracy as in \eqref{eq:final_error}. We note that when $\kappa = \order(1)$, our algorithm is orderwise faster than the AltProj algorithm in \cite{netrapalli2014non}, which has running time $\mathcal{O}(r^2d^2\log(1/\varepsilon))$. Moreover, our algorithm only requires computing one singular value decomposition.
\end{remark}
\vskip .05in
\begin{remark}[Robustness]
	Assuming $\kappa = \order(1)$, our algorithm can tolerate corruption at a sparsity level up to $\alpha = \order(1/(\mu r \sqrt{r}))$. This is worse by a factor $\sqrt{r}$ compared to the optimal statistical guarantee $1/(\mu r)$ obtained in \cite{chen2011LSarxiv,hsu2011robust,netrapalli2014non}. This looseness is a consequence of the condition for $ (\PIU_0, \PIV_0) $ in Theorem~\ref{thm:convergence}. Nevertheless, when $ \mu r = \order(1) $, our algorithm can tolerate a constant $ \alpha $ fraction of corruptions.
	
	Notably, we show that gradient descent works in the case of $\alpha = \order(1/(\mu r))$ if initialization is sufficiently close. Accordingly, to provide an algorithm with optimal robustness, it is straightforward to use a more complicated initial method such as AltProj that can tolerate $1/(\mu r)$ fraction of corruptions while satisfying our initial condition. As gradient descent provides a simple and efficient way to successively refine the estimation, such combination still gives a better running time than prior arts.
	%While slower than Algorithm \ref{alg:rpca}, this method still has significantly improved running time compared to the state of the art, and moreover, it has (optimal) provable performance guarantees.
\end{remark}

\subsection{Analysis of Algorithm \ref{alg:rpca_partial}}\label{sec:rpca_partial}
We now move to the guarantees of Algorithm \ref{alg:rpca_partial}. We show here that not only can we handle partial observations, but in fact subsampling the data in the fully observed case can significantly reduce the time complexity from the guarantees given in the previous section without sacrificing robustness. In particular, for smaller values of $r$, the complexity of Algorithm~\ref{alg:rpca_partial} has {\em near linear dependence on the dimension} $d$, instead of quadratic.

In the following discussion, we let $d := \max\{d_1, d_2\}$. The next two results, proved in Sections \ref{proof:thm:initial_partial} and \ref{proof:thm:convergence_partial}, control the quality of the initialization step, and then the gradient iterations. 

\begin{theorem}[Initialization, partial observations] \label{thm:initial_partial}
	Suppose the observed indices $\ObserveSet$ follow the Bernoulli model given in \eqref{eq:Y}. Consider the pair $(\PIU_0,\PIV_0)$ produced in the first phase of Algorithm \ref{alg:rpca_partial}. There exist constants $\{c_i\}_{i=1}^3$ such that for any $\epsilon \in (0, \sqrt{r}/(8c_1\kappa))$, if 
	\begin{equation} \label{eq:condition_thm_initial_partial}
	\alpha \leq \frac{1}{64\kappa \mu r}, ~~ p \geq c_2\left(\frac{\mu r^2}{\epsilon^2} + \frac{1}{\alpha}\right)\frac{\log d}{d_1 \wedge d_2}, 
	\end{equation}
	then we have 
	\[
	d(\PIU_0, \PIV_0; \OptU, \OptV) \leq 51\sqrt{\kappa} \alpha \mu r\sqrt{r} \sqrt{\sigma_1^*} + 7c_1\epsilon \sqrt{\kappa \sigma_1^*},
	\]
	with probability at least $1 - c_3d^{-1}$.
\end{theorem}

%The next theorem, proved in Section \ref{proof:thm:convergence_partial}, establishes the convergence guarantee for the gradient iterations in Algorithm \ref{alg:rpca_partial}. Similar to Theorem \ref{thm:convergence}, it ensures that the gradient iterates enjoy linear convergence to a global optimum when the initial point is sufficiently close to the optimal manifold.
\begin{theorem}[Convergence, partial observations] \label{thm:convergence_partial}
	Suppose the observed indices $\ObserveSet$ follow the Bernoulli model given in \eqref{eq:Y}. Consider the second phase of Algorithm \ref{alg:rpca_partial}. Suppose we choose $\gamma = 3$, and $\eta = c/(\mu r\sigma_1^*)$ for a sufficiently small constant $c$. There exist constants $\{c_i\}_{i=1}^4$ such that if 
	\begin{equation} \label{eq:condition_thm_convergence_partial}
	\alpha \leq \frac{c_1}{\kappa^2 \mu r} ~~\text{and}~~ p \geq c_2 \frac{\kappa^4 \mu ^2 r^2\log d}{d_1 \wedge d_2},
	\end{equation}
	then with probability at least $1 - c_3d^{-1}$, the iterates $\{(\PIU_t,\PIV_t)\}_{t = 0}^{\infty}$ satisfy
	\[
	d^2(U_t, V_t; \OptU, \OptV) \leq \left(1 - \frac{c}{64\mu r\kappa}\right)^td^2(U_0,V_0;\OptU, \OptV)
	\]
	for all $(\PIU_0, \PIV_0) \in \Neighbor{c_4\sqrt{\sigma^*_r/\kappa}}$.
\end{theorem}
The above result ensures linear convergence to $(\OptU,\OptV)$ (up to rotation) even when the gradient iterations are computed using partial observations. Note that setting $p = 1$ recovers Theorem \ref{thm:convergence} up to an additional factor $\mu r$ in the contraction factor. For achieving $\varepsilon$ relative accuracy, now we need $\order(\mu r \kappa\log(1/\varepsilon))$ iterations.

Putting Theorems \ref{thm:initial_partial} and \ref{thm:convergence_partial} together, we have the following overall guarantee, proved in Section \ref{proof:cor:rpca_partial}, for Algorithm \ref{alg:rpca_partial}.
\begin{corollary} \label{cor:rpca_partial}
	Suppose that
	\[
	\alpha \leq c\min\left\{\frac{1}{\mu\sqrt{\kappa r}^3}, \frac{1}{\mu \kappa^2r}\right\}, ~~ p \geq c'\frac{\kappa^4\mu^2 r^2\log d}{d_1 \wedge d_2},
	\] 
	for some constants $c, c'$. With probability at least $1 - \order(d^{-1})$, for any $\varepsilon \in (0,1)$, Algorithm \ref{alg:rpca_partial} with $T = \order(\mu r\kappa \log (1/\varepsilon))$ outputs a pair $(\PIU_{T}, \PIV_{T})$ that satisfies
	\begin{equation} \label{eq:final_error_partial}
	\frobnorm{\PIU_{T}\PIV_{T}^{\top} - \MStar} \leq \varepsilon\cdot\sigma_r^*.
	\end{equation}
\end{corollary}
This result shows that partial observations do not compromise robustness to sparse corruptions: as long as the observation probability $ p $ satisfies the condition in Corollary~\ref{cor:rpca_partial}, Algorithm \ref{alg:rpca_partial} enjoys the same robustness guarantees as the method using all entries. Below we provide two remarks on the sample and time complexity. For simplicity, we assume $d_1 = d_2 = d$, $\kappa = \order(1)$.
\vskip .05in
\begin{remark}[Sample complexity and matrix completion] Using the lower bound on $p$, it is sufficient to have $\order(\mu^2 r^2 d\log d)$  observed entries. In the special case $S^* = 0$, our partial observation model is equivalent to the model of exact matrix completion (see, e.g., \cite{candes2010power}). We note that our sample complexity (i.e., observations needed) matches that of completing a positive semidefinite (PSD) matrix by gradient descent as shown in \cite{chen2015fast}, and is better than the non-convex matrix completion algorithms in \cite{jain2013low} and \cite{sun2015guaranteed}. Accordingly, our result reveals the important fact that we can obtain robustness in matrix completion without deterioration of our statistical guarantees. It is known that that any algorithm for solving exact matrix completion must have sample size $\Omega(\mu rd\log d)$~\cite{candes2010power}, and a nearly tight upper bound $O(\mu rd\log^2d)$ is obtained in \cite{DBLP:journals/tit/Chen15} by convex relaxation. While sub-optimal by a factor $\mu r$, our algorithm is much faster than convex relaxation as shown below.
\end{remark}
\vskip .05in
\begin{remark}[Time complexity]
	Our sparse estimator on the sparse matrix with support $\Phi$ can be implemented  via partial quick sort with running time $\order({pd^2\log(\alpha p d)})$. Computing the gradient in each step involves the two terms in the objective function \eqref{eq:opt_partial}. Computing the gradient of the first term $\widetilde{\mathcal{L}}$ takes time $\order(r|\Phi|)$, whereas the second term takes time $\order(r^2d)$. In the initialization phase, performing rank-$r$ SVD on a sparse matrix with support $\Phi$ can be done in time $\order(r|\Phi|)$. We conclude that when $|\Phi| = \order(\mu^2 r^2 d\log d)$, Algorithm \ref{alg:rpca_partial} achieves the error bound \eqref{eq:final_error_partial} with running time $\order(\mu^3r^4d\log d \log(1/\varepsilon))$. Therefore, in the small rank setting with $r \ll d^{1/3}$, even when full observations are given, it is better to use Algorithm \ref{alg:rpca_partial} by subsampling the entries of $ Y $.
\end{remark}

%%%%%%%%%%%%%%%%%%%%%%%%%%%%%%%%%%%%%%%%%%%%%%%%%%%%%%%%%%%%%%%%%%%%%%

\section{Numerical Results} \label{sec:num_results}

In this section, we provide numerical results and compare the proposed algorithms with existing methods, including the inexact augmented lagrange multiplier (IALM) approach \cite{LinChenMa10} for solving the convex relaxation \eqref{eq:convex_rpca}  and the alternating projection (AltProj) algorithm proposed in \cite{netrapalli2014non}. All algorithms are implemented in MATLAB \footnote{Our code is available at \url{https://www.yixinyang.org/code/RPCA_GD.zip}.}, and the codes for existing algorithms are obtained from their authors. SVD computation in all algorithms uses the PROPACK library.\footnote{\url{http://sun.stanford.edu/~rmunk/PROPACK/}}  We ran all simulations on a machine with Intel 32-core Xeon (E5-2699) 2.3GHz with 240GB RAM.

%For the sparse estimation in our algorithm (line 6 in Algorithm \ref{alg:rpca}), we used the public mex function\footnote{\url{http://www.mathworks.com/matlabcentral/fileexchange/23576-min-max-selection}} which is based on partial quick sort.

\subsection{Synthetic Datasets}

 We generate a squared data matrix $Y = \MStar + \SStar \in \mathbb{R}^{d \times d}$ as follows. The low-rank part $\MStar$ is given by $M^* = A B^\top$, where $A,B \in \mathbb{R}^{d \times r}$ have entries drawn independently from a zero mean Gaussian distribution with variance $1/d$. For a given sparsity parameter $\alpha$, each entry of $\SStar$ is set to be nonzero with probability $\alpha$, and the values of the nonzero entries are sampled uniformly from $[-5r/d, 5r/d]$. 
\begin{figure}[t] 
	\centering
	\begin{subfigure}[t]{0.32\textwidth}
		\centering
		\includegraphics[width=\textwidth]{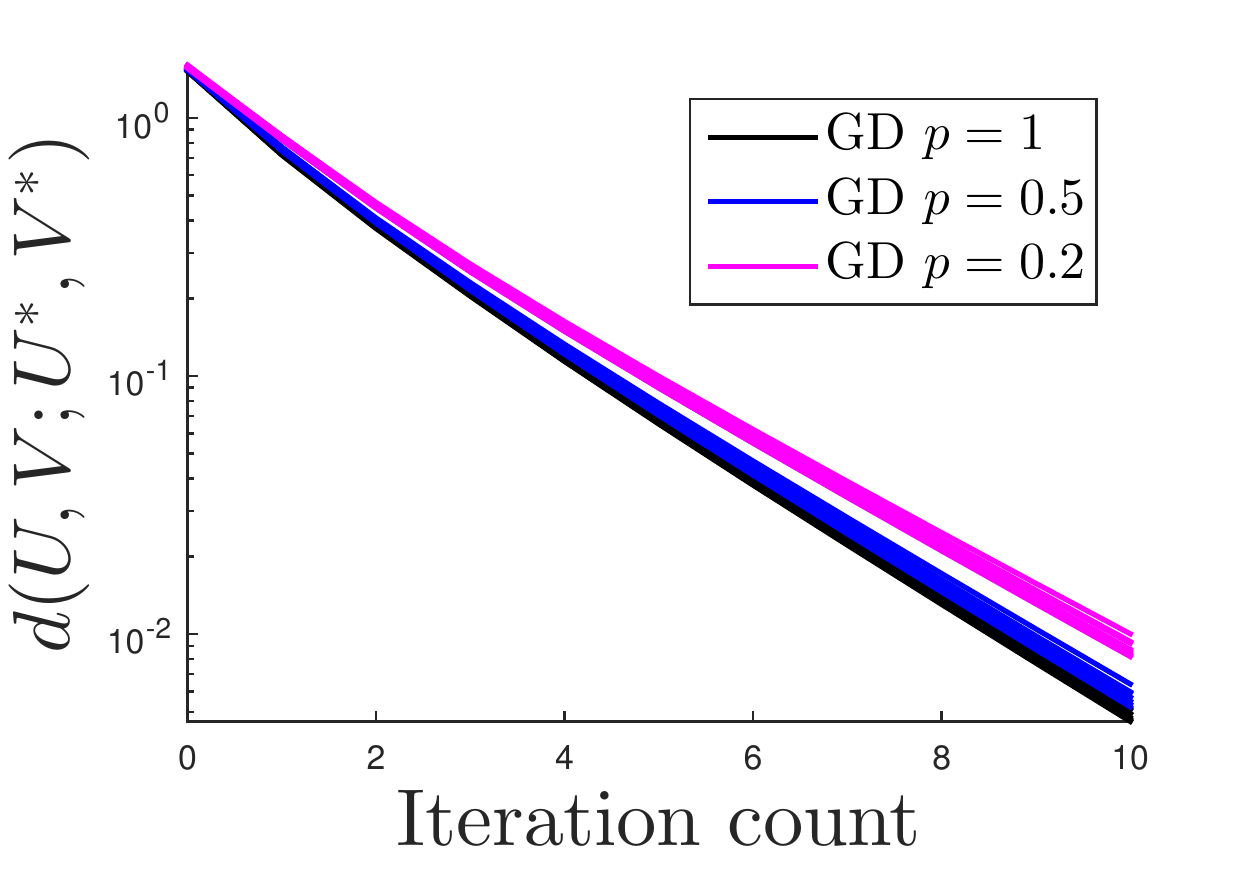}
		\caption{}
		\label{fig:conv_rate}
	\end{subfigure}%
	~ 
	\begin{subfigure}[t]{0.32\textwidth}
		\centering
		\includegraphics[width=\textwidth]{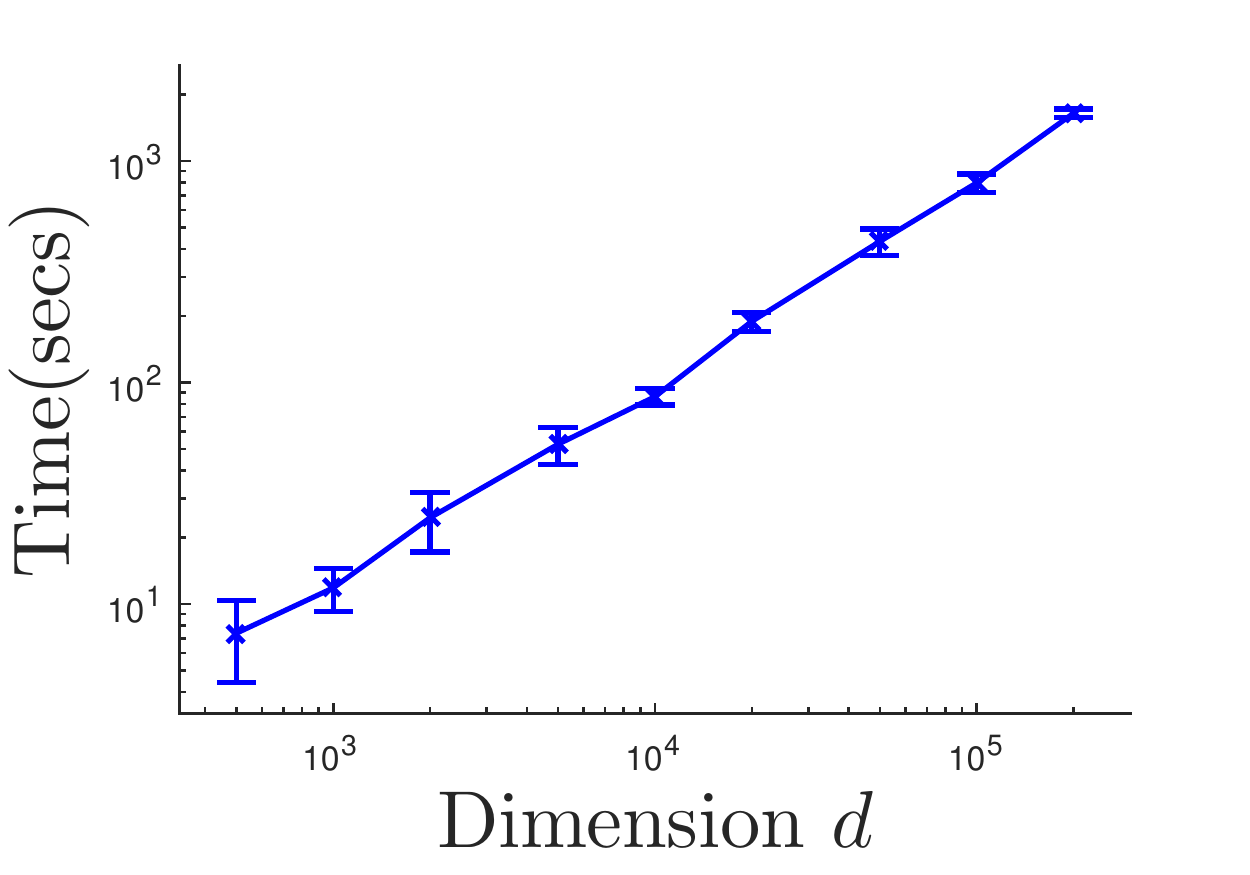}
		\caption{}
		\label{fig:time_scale}
	\end{subfigure}
	~
	\begin{subfigure}[t]{0.32\textwidth}
		\centering
		\includegraphics[width=\textwidth]{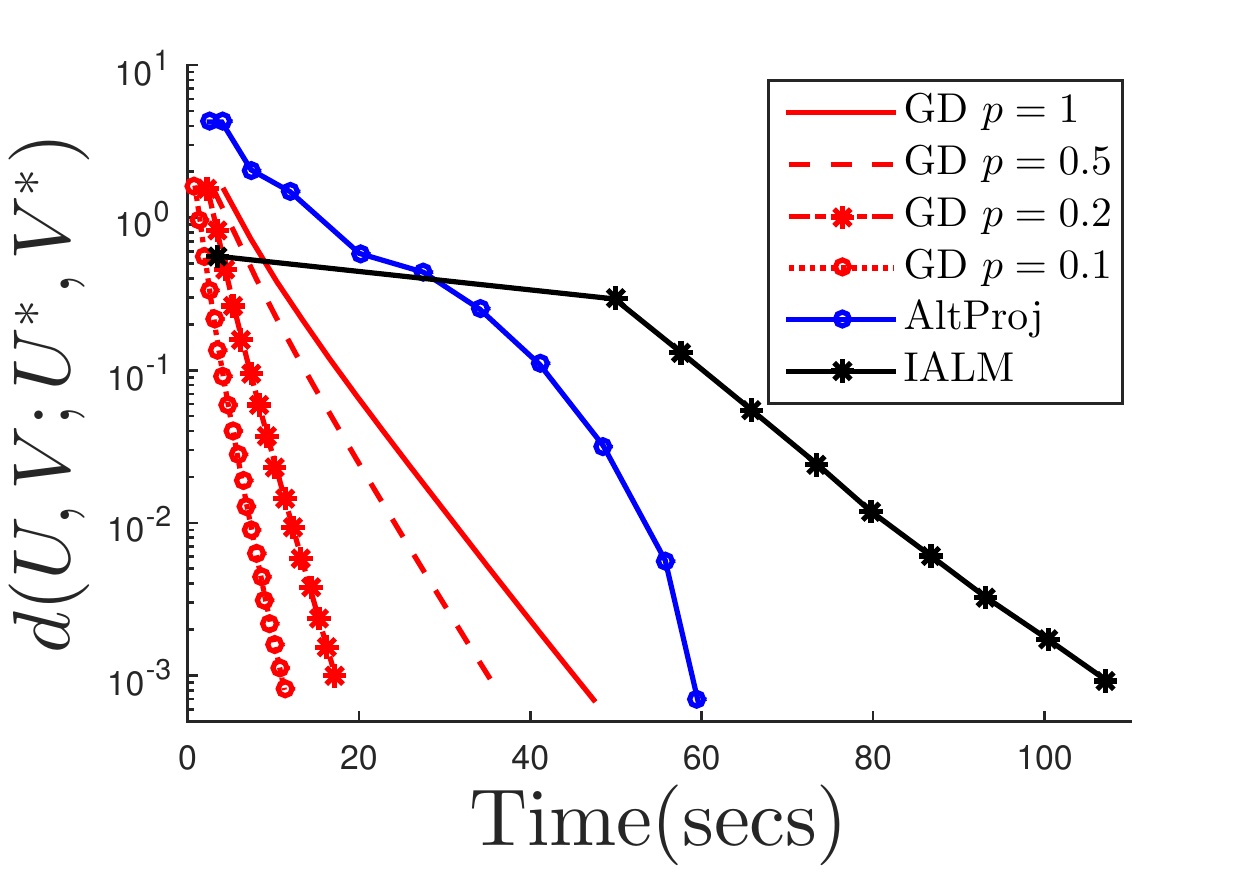}
		\caption{}
		\label{fig:real_time}
	\end{subfigure}
	\caption{Results on synthetic data. (a) Plot of log estimation error versus number of iterations when using gradient descent (GD) with varying sub-sampling rate $p$. It is conducted using $d = 5000, r = 10, \alpha = 0.1$. (b) Plot of running time of GD versus dimension $d$ with $r = 10, \alpha = 0.1, p = 0.15r^2\log d/d$. The low-rank matrix is recovered in all instances, and the line has slope approximately one. (c) Plot of log estimation error versus running time for different algorithms in problem with $d = 5000, r = 10, \alpha = 0.1$.}
	\label{fig:exp1}
\end{figure}

The results are summarized in Figure \ref{fig:exp1}. Figure \ref{fig:conv_rate} shows the convergence of our algorithms for different random instances with different sub-sampling rate $p$ (note that $p = 1$ corresponds to the fully observed setting). As predicted by Theorems \ref{thm:convergence} and \ref{thm:convergence_partial},  our gradient method converges geometrically with a contraction factor nearly independent of $p$. Figure \ref{fig:time_scale} shows the running time of our algorithm with partially observed data. We see that the running time scales linearly with $d$, again consistent with the theory. We note that our algorithm is memory-efficient: in the large scale setting with $d = 2 \times 10^5$, using approximately $0.1\%$ entries is sufficient for the successful recovery. In contrast, AltProj and IALM are designed to manipulate the entire matrix with $d^2 = 4 \times 10^{10}$ entries, which is prohibitive on a single machine. Figure \ref{fig:real_time} compares our algorithms with AltProj and IALM by showing reconstruction error versus real running time. Our algorithm requires significantly less computation to achieve the same accuracy level, and using only a subset of the entries provides additional speed-up.

\begin{figure}
	\centering
	\begin{tabular}{c}		
		\includegraphics[scale=0.75, clip, trim = 0 0 0 0]{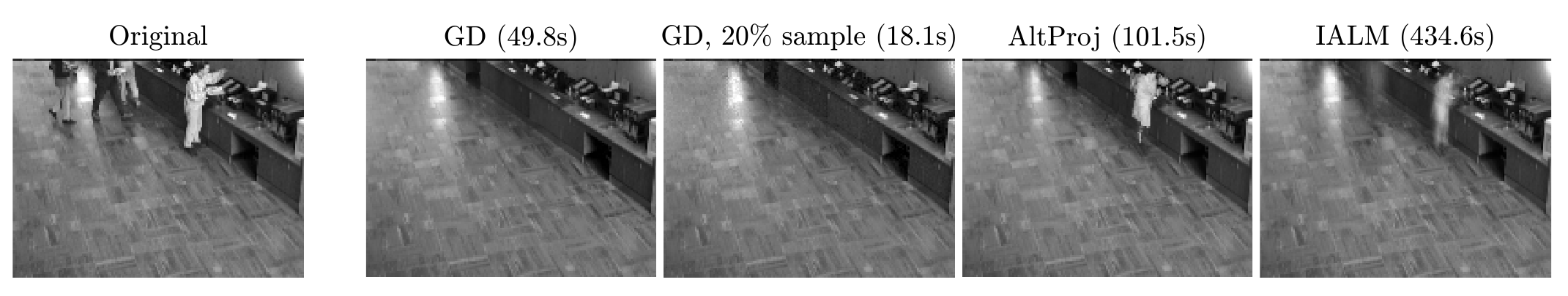} \\
		\includegraphics[scale=0.75, clip, trim = 0 0 0 0]{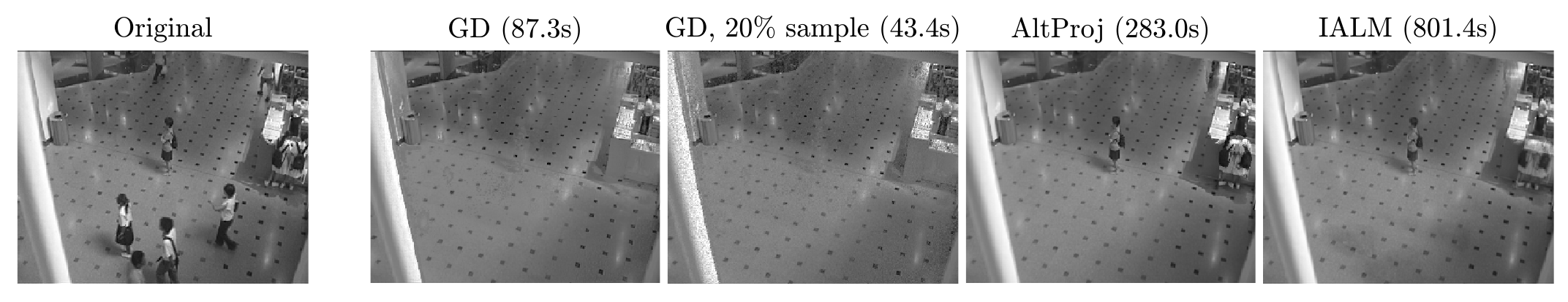}
	\end{tabular} 
	\caption{Foreground-background separation in \emph{Restaurant} and \emph{ShoppingMall} videos. In each line, the leftmost image is an original frame, and the other four are the separated background obtained from our algorithms with $p = 1$, $p = 0.2$, AltProj, and IALM. The running time required by each algorithm is shown in the title.}
	\label{fig:exp2}
\end{figure}

\subsection{Foreground-background Separation}

We apply our method to the task of foreground-background (FB) separation in a video. We use two public benchmarks, the \emph{Restaurant} and \emph{ShoppingMall} datasets.\footnote{\url{http://perception.i2r.a-star.edu.sg/bk_model/bk_index.html}} Each dataset contains a video with static background. By vectorizing and stacking the frames as columns of a matrix $Y$, the FB separation problem can be cast as RPCA, where the static background corresponds to a low rank matrix $M^*$ with identical columns, and the moving objects in the video can be modeled as sparse corruptions $S^*$. Figure \ref{fig:exp2} shows the output of different algorithms on two frames from the dataset. Our algorithms require significantly less running time than both AltProj and IALM. Moreover, even with 20\% sub-sampling, our methods still appear to achieve better separation quality (note that in each of the frames our algorithms remove a person that is not identified by the other algorithms). 

Figure \ref{fig:exp3} shows recovery results for several more frames. Again, our algorithms enjoy better running time and outperform AltProj and IALM in separating persons from the background images. In Appendix \ref{sec:fb_more}, we describe the detailed parameter settings for our algorithm.

\begin{figure}[!htp]
	\begin{center}
		\includegraphics[width=\textwidth]{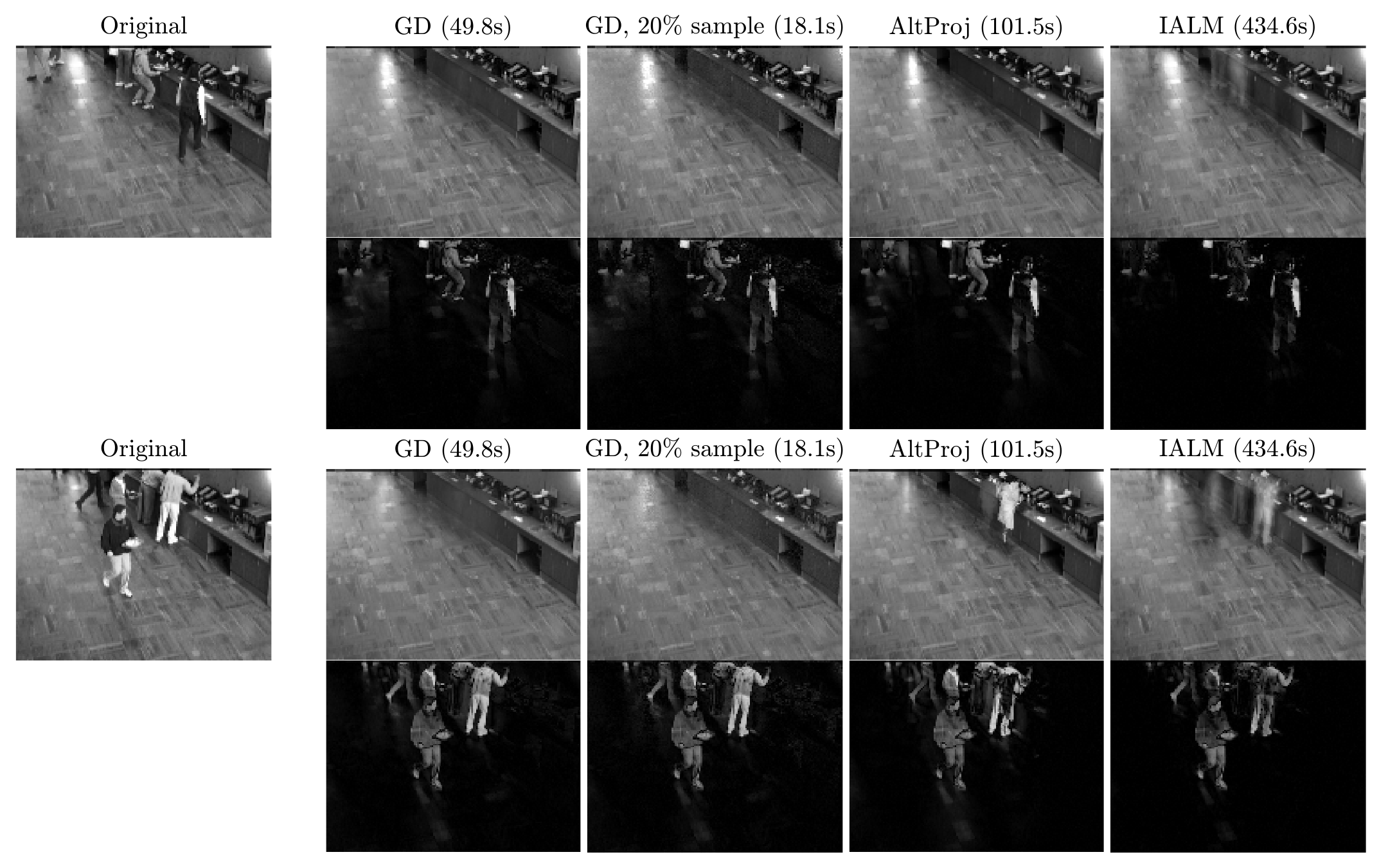}\\
		\includegraphics[width=\textwidth]{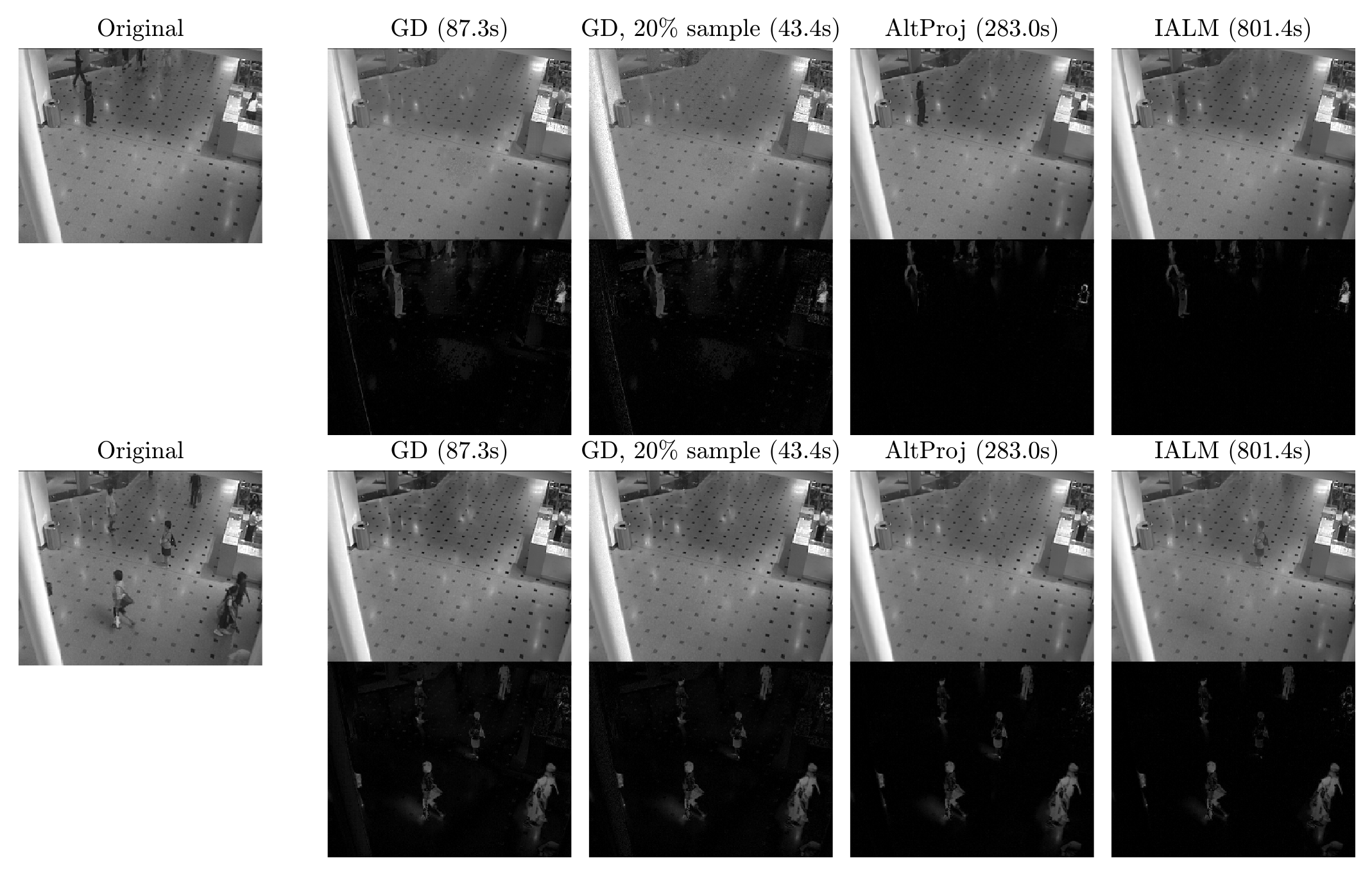}\\
		\caption{More results of FB separation in \emph{Restaurant} and \emph{ShoppingMall} videos. The leftmost images are original frames. The right four columns show the results from our algorithms with $p = 1, p = 0.2$, AltProj \cite{netrapalli2014non}, and IALM \cite{LinChenMa10}. The runtime of each algorithm is written in the title.}
		\label{fig:exp3}
	\end{center}
\end{figure}

\section{Proofs} \label{sec:proofs}

In this section we provide the proofs for our main theoretical results in Theorems~\ref{thm:initialization}--\ref{thm:convergence_partial} and Corollaries~\ref{cor:rpca}--\ref{cor:rpca_partial}.
\subsection{Proof of Theorem \ref{thm:initialization}} \label{proof:thm:initialization}

Let $\widebar{Y} := Y - \InitS$. As $Y = \MStar + \SStar$, we have $\widebar{Y} - \MStar = \SStar - \InitS$. We obtain $\widebar{Y} - \MStar \in \mathcal{S}_{2\alpha}$ because $\SStar, \InitS \in \mathcal{S}_{\alpha}$. 

We claim that $\infnorm{\widebar{Y} - \MStar} \leq 2\infnorm{\MStar}$. Denote the support of $\SStar, \InitS$ by $\Omega^*$ and $\Omega$ respectively. Since $\widebar{Y} - \MStar$ is supported on $\Omega \cup \Omega^*$, to prove the claim it suffices to consider the following three cases. 
\begin{itemize}
	\item For $(i,j) \in \Omega^*\cap \Omega$, due to rule of sparse estimation, we have $\entry{(\SStar - \InitS)}{i}{j} = 0$.
	\item For $(i,j) \in \Omega^* \setminus \Omega$, we must have $|\entry{\SStar}{i}{j}| \leq 2\infnorm{\MStar}$. Otherwise, we have $|\entry{Y}{i}{j}| = |\entry{(\SStar + \MStar)}{i}{j}| > \infnorm{\MStar}$. So $|\entry{Y}{i}{j}|$ is larger than any uncorrupted entries in its row and column. Since there are at most $\alpha$ fraction corruptions per row and column, we have $\entry{Y}{i}{j} \in \Omega$, which violates the prior condition $(i,j) \in \Omega^* \setminus \Omega$.
	\item For the last case $(i,j) \in \Omega \setminus \Omega^*$, since $\entry{(\InitS)}{i}{j} = \entry{\MStar}{i}{j}$, trivially we have $|\entry{(\InitS)}{i}{j}| \leq \infnorm{\MStar}$.
\end{itemize}

The following result, proved in  Section \ref{proof:lem:op_to_inf}, relates the operator norm of $\widebar{Y} - \MStar$ to its infinite norm.
\begin{lemma} \label{lem:op_to_inf}
	For any matrix $A \in \real^{d_1\times d_2}$ that belongs to $\Sparset{\alpha}$ given in \eqref{eq:S_alpha}, we have
	\[
	\opnorm{A} \leq   \alpha\sqrt{d_1 d_2}\infnorm{A}.
	\]
\end{lemma}

We thus obtain
\begin{equation} \label{eq:tmp_opnorm_bound}
\opnorm{\widebar{Y} - \MStar} \leq 2\alpha \sqrt{d_1d_2} \infnorm{\widebar{Y} - \MStar} \leq  4\alpha \sqrt{d_1d_2} \infnorm{\MStar} = 4\alpha \mu r \sigma_1^*.
\end{equation}
In the last step, we use the fact that $\MStar$ satisfies the $\mu$-incoherent condition, which leads to
\begin{equation} \label{eq:top_element}
\infnorm{\MStar} \leq \opnorm{\MStar}\twoinfnorm{\PIL^*}\twoinfnorm{\PIR^*} \leq \frac{\mu r}{\sqrt{d_1d_2}}\opnorm{\MStar}.
\end{equation}

We denote the $i$-th largest singular value of $\widebar{Y}$ by $\sigma_i$. By Weyl's theorem, we have $\abs{\sigma_i^* - \sigma_i} \leq \opnorm{\widebar{Y}-\MStar}$ for all $i \in [d_1\wedge d_2]$. Since $\sigma_{r+1}^* = 0$, we have $\sigma_{r+1} \leq \opnorm{\widebar{Y}-\MStar}$. Recall that $\PIU_0\PIV_0^{\top}$ is the best rank $r$ approximation of $\widebar{Y}$. Accordingly, we have
\begin{align*}
\opnorm{\PIU_0\PIV_0^{\top} - \MStar} & \leq \opnorm{\PIU_0\PIV_0^{\top} - \widebar{Y}} + \opnorm{ \widebar{Y} - \MStar} \\
& = \sigma_{r+1} + \opnorm{ \widebar{Y} - \MStar} \leq 2\opnorm{\widebar{Y} - \MStar} \leq 8\alpha\mu r \sigma_1^*.
\end{align*}

Under condition $\alpha \mu r \leq \frac{1}{16\kappa}$, we obtain $\opnorm{\PIU_0\PIV_0^{\top} - \MStar} \leq \frac{1}{2}\sigma_r^*$. Applying Lemma 5.14 in \cite{tu2015low} (we provide it as Lemma \ref{lem:tu_lemma} for the sake of completeness), we obtain
\[
d^2(\PIU_0, \PIV_0; \OptU, \OptV) \leq \frac{2}{\sqrt{2} - 1}\frac{\frobnorm{\PIU_0\PIV_0^{\top} - \MStar}^2}{\sigma_r^*} \leq \frac{10r\opnorm{\PIU_0\PIV_0^{\top} - \MStar}^2}{\sigma_r^*}.
\]
Plugging the upper bound of $\opnorm{\PIU_0\PIV_0^{\top} - \MStar}$ into the above inequality completes the proof.

\subsection{Proof of Theorem \ref{thm:convergence}} \label{proof:thm:convergence}
We essentially follow the general framework developed in \cite{chen2015fast} for analyzing the behaviors of gradient descent in factorized low-rank optimization. But it is worth to note that \cite{chen2015fast} only studies the symmetric and positive semidefinite setting, while we avoid such constraint on $\MStar$. The techniques for analyzing general asymmetric matrix in factorized space is inspired by the recent work \cite{tu2015low} on solving low-rank matrix equations. In our setting, the technical challenge is to verify the local descent condition of the loss function \eqref{eq:opt}, which not only has a bilinear dependence on $\PIU$ and $\PIV$, but also involves our sparse estimator \eqref{eq:sparse_estimator}. 

We begin with some notations. Define the equivalent set of optimal solution as
\begin{equation} \label{eq:Eq_set}
\OptSet(\MStar) := \left\{(A,B) \in \real^{d_1 \times r} \times \real^{d_2 \times r} ~\big|~ A = \PIL^*\SVMat^{*1/2}\RotateM,  B = \PIR^*\SVMat^{*1/2}\RotateM, ~\text{where}~\RotateM \in \RotateSet{r} \right\}.
\end{equation}

Given $(\PIU_0, \PIV_0) \in \Neighbor{c_2\sqrt{\sigma^*_r/\kappa}}$, by \eqref{eq:error_relation}, we have $\opnorm{\PIU_0\PIV_0^{\top} - \MStar} \leq \frac{1}{2}\sigma_r^*$ when $c_2$ is sufficiently small. By Weyl's theorem We thus have \[ \sqrt{\sigma_1^*/2}  \leq \opnorm{\PIU_0} \leq \sqrt{3\sigma_1^*/2}, ~\text{and}~  \sqrt{\sigma_1^*/2} \leq \opnorm{\PIV_0} \leq \sqrt{3\sigma_1^*/2}.
\] 
As a result, for $\USet, \VSet$ constructed according to \eqref{eq:set}, we have 
\begin{equation} \label{eq:uvcondition}
\OptSet(\MStar) \subseteq \USet \times \VSet, ~\text{and}~  \USet \subseteq \widebar{\USet},  \VSet \subseteq \widebar{\VSet},
\end{equation}
where 
\begin{equation*}
\widebar{\USet} := \left\{A \in \real^{d_1 \times r} ~\big|~ \twoinfnorm{A} \leq \sqrt{\frac{3\mu r\sigma_1^*}{d_1}}\right\}, ~
\widebar{\VSet} := \left\{A \in \real^{d_2 \times r}~\big|~ \twoinfnorm{A} \leq \sqrt{\frac{3\mu r\sigma_1^*}{d_2}}\right\}.
\end{equation*}

We let 
\begin{equation} \label{eq:G}
\mathcal{G}(\PIU, \PIV) : = \frac{1}{8}\frobnorm{\PIU^{\top}\PIU - \PIV^{\top}\PIV}^2.
\end{equation}
For $\Loss(\PIU, \PIV; \PIS)$, we denote the gradient with respect to $M$ by $\nabla_{M}\Loss(\PIU,\PIV;\PIS)$, i.e. $\nabla_{M}\Loss(\PIU,\PIV;\PIS) = \PIU\PIV^{\top} + S - Y$.

The  local descent property is implied by combining the following two results, which are proved in Section \ref{proof:lem:local_descent_L} and \ref{proof:lem:local_descent_g} respectively.
\begin{lemma}[Local descent property of $\Loss$] \label{lem:local_descent_L}
	Suppose $\USet, \VSet$ satisfy \eqref{eq:uvcondition}. For any $(U,V) \in (\USet \times \VSet) \cap \mathbb{B}_2(\sqrt{\sigma_1^*})$, we let $S = \UnifSparse{Y - \PIU\PIV^{\top}}{\gamma \alpha}$, where we choose $\gamma = 2$. Then we have that for $(\UniqueU, \UniqueV) \in \argmin_{(A, B) \in \OptSet(\MStar)} \frobnorm{\PIU - A}^2 + \frobnorm{\PIV - B}^2$ and  $\beta > 0$,
	\begin{align}\label{eq:local_descent}
	& \trinprod{\nabla_{M} \Loss(\PIU, \PIV; S)}{\PIU \PIV^{\top} - \UniqueU \UniqueV^{\top} + \DeltaU \DeltaV^{\top}}  \geq \frobnorm{\PIU \PIV^{\top} - \UniqueU \UniqueV^{\top}}^2 -  \nu\sigma_1^*\delta - 3\sqrt{\sigma_1^*\delta^3}.
	\end{align}
	Here $\DeltaU := \PIU - \UniqueU$, $\DeltaV := \PIV - \UniqueV$, $\delta := \frobnorm{\DeltaU}^2 + \frobnorm{\DeltaV}^2$, and $\nu := 9(\beta+6)\alpha \mu r + 5\beta^{-1}$.
\end{lemma}  

\vskip 0.1in
\begin{lemma}[Local descent property of $\mathcal{G}$] \label{lem:local_descent_g}
	For any $(\PIU,\PIV) \in \mathbb{B}_2(\sqrt{\sigma_r^*})$ and 
	\[
	(\UniqueU, \UniqueV) \in \argmin_{(A, B) \in \OptSet(\MStar)} \frobnorm{\PIU - A}^2 + \frobnorm{\PIV - B}^2,
	\] 
	we have
	\begin{align*}
	& \trinprod{\nabla_{\PIU} \mathcal{G}(\PIU, \PIV)}{U -  \UniqueU} + \trinprod{\nabla_{\PIV} \mathcal{G}(\PIU, \PIV)}{V -  \UniqueV} \\
	& \geq \frac{1}{8}\frobnorm{\PIU^{\top}\PIU - \PIV^{\top}\PIV}^2 + \frac{1}{8}\sigma_r^*\delta -  \sqrt{\frac{\sigma_1^*\delta^3}{2}} -\frac{1}{2}\frobnorm{\PIU \PIV^{\top} - \UniqueU\UniqueV^{\top}}^2,
	\end{align*}
	where $\delta$ is defined according to Lemma \ref{lem:local_descent_L}.
\end{lemma}

As another key ingredient, we establish the following smoothness condition, proved in Section \ref{proof:lem:smoothness}, which indicates that the Frobenius norm of gradient decreases as $(\PIU, \PIV)$ approaches the optimal manifold.

\begin{lemma}[Smoothness] \label{lem:smoothness}
	For any $(\PIU, \PIV) \in \mathbb{B}_2(\sqrt{\sigma_1^*})$, we let $S = \UnifSparse{Y - \PIU\PIV^{\top}}{\gamma\alpha}$, where we choose $\gamma = 2$. We have that
	\begin{equation} \label{eq:smooth_L}
	\frobnorm{\nabla_{M}\Loss(\PIU,\PIV;\PIS)}^2 \leq  6\frobnorm{\PIU\PIV^{\top} - \MStar}^2,
	\end{equation}
	and
	\begin{equation} \label{eq:smooth_G}
	\frobnorm{\nabla_{\PIU} \mathcal{G}(\PIU, \PIV)}^2 + \frobnorm{\nabla_{\PIV} \mathcal{G}(\PIU, \PIV)}^2 \leq 2\sigma_1^*\frobnorm{\PIU^{\top}\PIU - \PIV^{\top}\PIV}^2.
	\end{equation}
\end{lemma}
With the above results in hand, we are ready to prove Theorem \ref{thm:convergence}.
\begin{proof}[Proof of Theorem \ref{thm:convergence}.] 
	We use shorthands 
	\[
	\delta_t := d^2(U_t, V_t; \OptU, \OptV), ~\Loss_t := \Loss(\PIU_t,\PIV_t; \PIS_t),~ \text{and}~~ \mathcal{G}_t := \mathcal{G}(\PIU_t,\PIV_t).
	\] 
	For $(\PIU_t, \PIV_t)$, let $(\UniqueU^{t}, \UniqueV^{t}) := \argmin_{(A, B) \in \OptSet(\MStar)} \frobnorm{\PIU_t - A}^2 + \frobnorm{\PIV_t - B}^2$. Define $\Delta_{\PIU}^t := \PIU_t - \UniqueU^{t}$, $\Delta_{\PIV}^t := \PIV_t - \UniqueV^{t}$.
	
	We prove Theorem \ref{thm:convergence} by induction. It is sufficient to consider one step of the iteration. For any $t \geq 0$, under the induction hypothesis $(\PIU_t, \PIV_t) \in \Neighbor{c_2\sqrt{\sigma^*_r/\kappa}}$. We find that
	\begin{align}
	\delta_{t+1} & \leq \frobnorm{U_{t+1} - \UniqueU^t}^2 + \frobnorm{V_{t+1} - \UniqueV^t}^2 \notag\\
	& \leq \frobnorm{U_t - \eta \nabla_{\PIU}\Loss_t - \eta\nabla_{\PIU}\mathcal{G}_t - \UniqueU^t}^2  + \frobnorm{V_t - \eta \nabla_{\PIV}\Loss_t - \eta\nabla_{\PIV}\mathcal{G}_t - \UniqueV^t}^2 \notag \\
	& \leq \delta_t - 2\eta\underbrace{\trinprod{\nabla_{\PIU}\Loss_t + \nabla_{\PIU}\mathcal{G}_t}{\PIU_t - \UniqueU^t}}_{W_1} - 2\eta\underbrace{\trinprod{\nabla_{\PIV}\Loss_t + \nabla_{\PIV}\mathcal{G}_t}{\PIV_t - \UniqueV^t}}_{W_2} \notag\\
	& ~~~~~~~ + \eta^2 \underbrace{\frobnorm{\nabla_{\PIU}\Loss_t + \nabla_{\PIU}\mathcal{G}_t}^2}_{W_3} + \eta^2\underbrace{\frobnorm{\nabla_{\PIV}\Loss_t + \nabla_{\PIV}\mathcal{G}_t}^2}_{W_4}, \label{eq:delta_relation}
	\end{align}
	where the second step follows from the non-expansion property of projection onto $\USet, \VSet$, which is implied by $\OptSet(\MStar) \subseteq \USet \times \VSet$ shown in \eqref{eq:uvcondition}. Since $\nabla_{\PIU}\Loss_t = \left[\nabla_{M}\Loss_t\right] \PIV$ and $\nabla_{\PIV}\Loss_t = \left[\nabla_{M}\Loss_t\right]^{\top} \PIU$, we have
	\[
	\trinprod{\nabla_{\PIU}\Loss_t}{\PIU_t - \UniqueU^t} + \trinprod{\nabla_{\PIV}\Loss_t}{\PIV_t - \UniqueV^t} = \trinprod{\nabla_{M} \Loss_t}{\PIU_t \PIV_t^{\top} - \UniqueU^t \UniqueV^{t\top} + \DeltaU^t \DeltaV^{t\top}}.
	\]
	
	Combining Lemma \ref{lem:local_descent_L} and \ref{lem:local_descent_g}, under condition $\delta_t < \sigma_r^*$, we have that
	\[
	W_1 + W_2 \geq \frac{1}{2}\frobnorm{\PIU_t \PIV_t^{\top} - \MStar}^2 +  \frac{1}{8}\frobnorm{\PIU_t^{\top}\PIU_t - \PIV_t^{\top}\PIV_t}^2 + \frac{1}{8}\sigma_r^*\delta_t -  \nu\sigma_1^*\delta_t - 4\sqrt{\sigma_1^*\delta_t^3}.
	\] 
	
	On the other hand, we have
	\begin{align*}
	W_3 + W_4 & \leq 2\frobnorm{\nabla_{\PIU}\Loss_t}^2 + 2\frobnorm{\nabla_{\PIU}\mathcal{G}_t}^2 +  2\frobnorm{\nabla_{\PIV}\Loss_t}^2 + 2\frobnorm{\nabla_{\PIV}\mathcal{G}_t}^2 \\
	& \leq 2(\opnorm{\PIU_t}^2 + \opnorm{\PIV_t}^2)\frobnorm{\nabla_M \Loss_t}^2 + 2\frobnorm{\nabla_{\PIU}\mathcal{G}_t}^2 + 2\frobnorm{\nabla_{\PIV}\mathcal{G}_t}^2 \\
	& \leq 36\sigma_1^*\frobnorm{\PIU_t\PIV_t^{\top} - \MStar}^2 +  4\sigma_1^*\frobnorm{\PIU_t^{\top}\PIU_t - \PIV_t^{\top}\PIV_t}^2,
	\end{align*}
	where the last step is implied by Lemma \ref{lem:smoothness} and the assumption $(\PIU_t, \PIV_t) \in \Neighbor{c_2\sqrt{\sigma^*_r/\kappa}}$ that leads to $\opnorm{\PIU_t} \leq \sqrt{3\sigma_1^*/2}, \opnorm{\PIV_t} \leq \sqrt{3\sigma_1^*/2}$.
	
	By the assumption $\eta = c/\sigma_1^*$ for any constant $c \leq 1/36$, we thus have
	\begin{align*}
	-2\eta(W_1 + W_2) + \eta^2(W_3 + W_4) \leq  - \frac{1}{4}\eta\sigma_r^*\delta_t + 2\eta  \nu\sigma_1^*\delta_t +  8\eta\sqrt{\sigma_1^*\delta_t^3}.
	\end{align*}
	In Lemma \ref{lem:local_descent_L}, choosing $\beta = 320\kappa$ and assuming $\alpha \lesssim 1/(\kappa^2\mu r)$, we can have $\nu \leq 1/(32\kappa)$. Assuming $\delta_t \lesssim \sigma_r^*/\kappa$ leads to $14\sqrt{\sigma_1^*\delta_t^3} \leq \frac{1}{16}\sigma_r^*\delta_t$. We thus obtain
	\begin{equation} \label{eq:one_step}
	\delta_{t+1} \leq \left(1 - \frac{\eta \sigma_r^*}{8}\right)\delta_t.
	\end{equation}
	Under initial condition $\delta_0 \lesssim \sigma_r^*/\kappa$, we obtain that such condition holds for all $t$ since estimation error decays geometrically after each iteration. Then applying \eqref{eq:one_step} for all iterations, we conclude that for all $t = 0,1,\ldots$,
	\[
	\delta_t \leq \left(1 - \frac{\eta \sigma_r^*}{8}\right)^t\delta_0.
	\]
\end{proof}

\subsection{Proof of Corollary \ref{cor:rpca}} \label{proof:cor:rpca}

We need $\alpha \lesssim \frac{1}{\kappa^2 \mu r}$ due to the condition of Theorem \ref{thm:convergence}. In order to ensure the linear convergence happens, it suffices to let the initial error shown in Theorem \ref{thm:initialization} be less than the corresponding condition in Theorem \ref{thm:convergence}. Accordingly, we need
\[
28 \sqrt{\kappa} \alpha \mu r\sqrt{r} \sqrt{\sigma_1^*} \lesssim \sqrt{\sigma^*_r/\kappa},
\]
which leads to $\alpha \lesssim \frac{1}{\mu \sqrt{r\kappa}^3}$.

Using the conclusion that gradient descent has linear convergence, choosing $T = \order(\kappa \log(1/\varepsilon))$, we have
\[
d^2(\PIU_{T}, \PIV_{T}; \OptU, \OptV) \leq \varepsilon^2 d^2(\PIU_{0}, \PIV_{0}; \OptU, \OptV) \lesssim \varepsilon^2 \frac{\sigma_r^*}{\kappa}.
\]
Finally, applying the relationship between $d(\PIU_{T}, \PIV_{T}; \OptU, \OptV)$ and $\frobnorm{\PIU_{T}\PIV_{T}^{\top} - \MStar}$ shown in \eqref{eq:error_relation}, we complete the proof.

\subsection{Proof of Theorem \ref{thm:initial_partial}} \label{proof:thm:initial_partial}

Let $\widebar{Y} := \frac{1}{p}(Y - \InitS)$. Similar to the proof of Theorem \ref{thm:initialization}, we first establish an upper bound on $\opnorm{\widebar{Y} - \MStar}$. We have that
\begin{equation} \label{eq:op_bound}
\opnorm{\widebar{Y} - \MStar} \leq \opnorm{\widebar{Y} - \frac{1}{p}\Proj_{\ObserveSet}\MStar} + \opnorm{\frac{1}{p}\Proj_{\ObserveSet}(\MStar) - \MStar}.
\end{equation}

For the first term, we have $\widebar{Y} - \frac{1}{p}\Proj_{\ObserveSet}\MStar = \frac{1}{p}(\Proj_{\ObserveSet}(\SStar) - \InitS)$ because $Y = \Proj_{\ObserveSet}(\MStar + \SStar)$.  Lemma \ref{lem:size_concentration} shows that under condition $p \gtrsim \frac{\log d}{\alpha (d_1 \wedge d_2)}$, there are at most $\frac{3}{2}p\alpha$-fraction nonzero entries in each row and column of $\Proj_{\ObserveSet}(\SStar)$ with high probability. Since $\InitS \in \Sparset{2p\alpha}$, we have 
\begin{equation} \label{eq:property_1}
\Proj_{\ObserveSet}(\SStar) - \InitS \in \Sparset{4p\alpha}.
\end{equation}
In addition, we prove below that 
\begin{equation} \label{eq:property_2}
\infnorm{\Proj_{\ObserveSet}(\SStar) - \InitS} \leq 2\infnorm{\MStar}.
\end{equation}
Denote the support of $\Proj_{\ObserveSet}(\SStar)$ and $\InitS$ by $\Omega_o^*$ and $\Omega$. For $(i,j) \in \Omega_o^* \cap \Omega$ and $(i,j) \in \Omega \setminus \Omega_o^*$, we have $\entry{(\Proj_{\ObserveSet}(\SStar) - \InitS)}{i}{j} = 0$ and $\entry{(\Proj_{\ObserveSet}(\SStar) - \InitS)}{i}{j} = -\entry{\MStar}{i}{j}$, respectively. To prove the claim, it remains to show that for $(i,j) \in \Omega_o^*\setminus \Omega$, $\abs{\entry{\SStar}{i}{j}} < 2\infnorm{\MStar}$. If this is not true, then we must have $\abs{\entry{Y}{i}{j}} > \infnorm{\MStar}$. Accordingly, $\abs{\entry{Y}{i}{j}}$ is larger than the magnitude of any uncorrupted entries in its row and column. Note that on the support $\Phi$, there are at most $\frac{3}{2}p\alpha$ corruptions per row and column, we have $(i,j) \in \Omega$, which violates our prior condition $(i,j) \in \Omega_o^*\setminus \Omega$. 

Using these two properties \eqref{eq:property_1}, \eqref{eq:property_2} and applying Lemma \ref{lem:op_to_inf}, we have
\begin{equation} \label{eq:opnorm_bound_1}
\opnorm{\widebar{Y}- \frac{1}{p}\Proj_{\ObserveSet}\MStar} \leq 4\alpha\sqrt{d_1d_2}\infnorm{\Proj_{\ObserveSet}(\SStar) - \InitS} \leq 8\alpha\sqrt{d_1d_2}\infnorm{\MStar} \leq 8\alpha \mu r \sigma_1^*,
\end{equation}
where the last step follow from \eqref{eq:top_element}.

For the second term in \eqref{eq:op_bound}, we use the following lemma proved in \cite{DBLP:journals/tit/Chen15}.
\begin{lemma}[Lemma 2 in \cite{DBLP:journals/tit/Chen15}] \label{lem:op_sampling}
	Suppose $A \in \real^{d_1 \times d_2}$ is a fixed matrix. We let $d := \max\{d_1, d_2\}$. There exists a constant $c$ such that with probability at least $1 - \order(d^{-1})$,
	\[
	\opnorm{\frac{1}{p}\Proj_{\ObserveSet}(A) - A} \leq c\left(\frac{\log d}{p}\infnorm{A} + \sqrt{\frac{\log d}{p}} \max\left\{\twoinfnorm{A}, \twoinfnorm{A^{\top}}\right\} \right).
	\]
\end{lemma}
Given the SVD $\MStar = \PIL^*\SVMat \PIR^{*\top}$, for any $i \in [d_1]$, we have
\[
\twonorm{\row{\MStar}{i}} = \twonorm{\row{\PIL^*}{i}\SVMat \PIR^{*\top}} \leq \sigma_1^*\twonorm{\row{\PIL^*}{i}} \leq \sigma_1^*\sqrt{\frac{\mu r}{d_1}}.
\]
We can bound $\twoinfnorm{M^{*\top}}$ similarly. Lemma \ref{lem:op_sampling} leads to
\begin{equation} \label{eq:opnorm_bound_2}
\opnorm{\frac{1}{p}\Proj_{\ObserveSet}(\MStar) - \MStar} \leq c'\frac{\log d}{p}\frac{\mu r}{\sqrt{d_1d_2}}\sigma_1^* + c'\sqrt{\frac{\log d}{p}}\sqrt{\frac{\mu r}{d_1 \wedge d_2}}\sigma_1^* \leq c' \epsilon \sigma_1^*/\sqrt{r}
\end{equation}
under condition $p \geq \frac{4\mu r^2\log d}{\epsilon^2(d_1 \wedge d_2)}$. 

Putting \eqref{eq:opnorm_bound_1} and \eqref{eq:opnorm_bound_2} together, we obtain
\[
\opnorm{\widebar{Y} - \MStar} \leq 8\alpha\mu r\sigma_1^* + c'\epsilon\sigma_1^*/\sqrt{r}.
\]
Then using the fact that $\PIU_0\PIV_0^{\top}$ is the best rank $r$ approximation of $\widebar{Y}$ and applying Wely's theorem (see the proof of Theorem \ref{thm:initialization} for a detailed argument), we have
\begin{align*}
\opnorm{\PIU_0\PIV_0^{\top} - \MStar} & \leq \opnorm{\PIU_0\PIV_0^{\top} - \widebar{Y}} + \opnorm{\widebar{Y} - \MStar} \\
& \leq 2\opnorm{\widebar{Y} - \MStar} \leq 16\alpha\mu r\sigma_1^* + 2c'\epsilon\sigma_1^*/\sqrt{r}
\end{align*}

Under our assumptions,  we have $16\alpha\mu r\sigma_1^* + 2c'\epsilon\sigma_1^*/\sqrt{r} \leq \frac{1}{2}\sigma_r^*$. Accordingly, Lemma \ref{lem:tu_lemma} gives
\[
d^2(\PIU_0, \PIV_0; \OptU, \OptV) \leq \frac{2}{\sqrt{2} - 1}\frac{\frobnorm{\PIU_0\PIV_0^{\top} - \MStar}^2}{\sigma_r^*} \leq \frac{10r\opnorm{\PIU_0\PIV_0^{\top} - \MStar}^2}{\sigma_r^*}.
\]
We complete the proof by combining the above two inequalities.

\subsection{Proof of Theorem \ref{thm:convergence_partial}} \label{proof:thm:convergence_partial}
In this section, we turn to prove Theorem \ref{thm:convergence_partial}. Similar to the proof of Theorem \ref{thm:convergence}, we rely on establishing the local descent and smoothness conditions. Compared to the full observation setting, we replace $\Loss$ by $\widetilde{\Loss}$ given in \eqref{eq:loss_partial}, while the regularization term $\widetilde{\mathcal{G}}(\PIU,\PIV) := \frac{1}{64}\frobnorm{\PIU^{\top}\PIU - \PIV^{\top}\PIV}^2$ merely differs from $\mathcal{G}(\PIU,\PIV)$ given in \eqref{eq:G} by a constant factor. It is thus sufficient to analyze the properties of $\widetilde{\Loss}$. 

Define $\OptSet(\MStar)$ according to \eqref{eq:Eq_set}. Under the initial condition, we still have
\begin{equation} \label{eq:uvcondition_partial}
\OptSet(\MStar) \subseteq \USet \times \VSet, ~\text{and}~  \USet \subseteq \widebar{\USet},  \VSet \subseteq \widebar{\VSet}.
\end{equation}

We prove the next two lemmas in Section \ref{proof:lem:local_descent_L_partial} and \ref{proof:lem:smoothness_partial} respectively. In both lemmas, for any $(\PIU, \PIV) \in \USet \times \VSet$, we use shorthands 
\[
(\UniqueU, \UniqueV) = \argmin_{(A, B) \in \OptSet(\MStar)} \frobnorm{\PIU - A}^2 + \frobnorm{\PIV - B}^2,
\]
$\DeltaU := \PIU - \UniqueU$, $\DeltaV := \PIV - \UniqueV$, and $\delta := \frobnorm{\DeltaU}^2 + \frobnorm{\DeltaV}^2$. Recall that $d := \max\{d_1, d_2\}$.
\begin{lemma}[Local descent property of $\widetilde{\Loss}$] \label{lem:local_descent_L_partial}
	Suppose $\USet, \VSet$ satisfy \eqref{eq:uvcondition_partial}. Suppose we let 
	\[
	S = \UnifSparse{\Proj_{\ObserveSet}\left(Y - \PIU\PIV^{\top}\right)}{\gamma p \alpha},
	\] 
	where we choose $\gamma = 3$. For any $\beta > 0$ and $\epsilon \in (0, \frac{1}{4})$, we define $\nu := (14\beta + 81)\alpha \mu r + 26\sqrt{\epsilon} + 18\beta^{-1}$. There exist constants $\{c_i\}_{i=1}^2$ such that if 
	\begin{equation} \label{eq:condition_p}
	p \geq c_1\left(\frac{\mu ^2 r^2}{\epsilon^2} + \frac{1}{\alpha} + 1\right)\frac{\log d}{d_1 \wedge d_2},
	\end{equation}
	then with probability at least $1 - c_2d^{-1}$,
	\begin{equation}\label{eq:local_descent_partial}
	\trinprod{\nabla_{M} \widetilde{\Loss}(\PIU, \PIV; S)}{\PIU \PIV^{\top} - \UniqueU \UniqueV^{\top} + \DeltaU \DeltaV^{\top}} 
	\geq \frac{3}{16}\frobnorm{\PIU \PIV^{\top} - \UniqueU \UniqueV^{\top}}^2 - \nu \sigma_1^*\delta - 10\sqrt{\sigma_1^*\delta^3} - 2\delta^2
	\end{equation}
	for all $(\PIU, \PIV) \in \left(\USet\times\VSet\right) \cap \Neighbor{\sqrt{\sigma_1^*}}$.
\end{lemma}

\begin{lemma}[Smoothness of $\widetilde{\Loss}$] \label{lem:smoothness_partial}
	Suppose $\USet, \VSet$ satisfy \eqref{eq:uvcondition_partial}. Suppose we let $S = \UnifSparse{\Proj_{\ObserveSet}\left(Y - \PIU\PIV^{\top}\right)}{\gamma\alpha p}$ for $\gamma = 3$. There exist constants $\{c_i\}_{i=1}^3$ such that for any $\epsilon \in (0,\frac{1}{4})$, when $p$ satisfies condition \eqref{eq:condition_p}, with probability at least $1 - c_2d^{-1}$, we have that for all $(\PIU, \PIV) \in (\USet \times \VSet) \cap \mathbb{B}_2(\sqrt{\sigma_1^*})$,
	\begin{equation} \label{eq:smooth_partial}
	\frobnorm{\nabla_{\PIU}\widetilde{\Loss}(\PIU,\PIV;\PIS)}^2 + \frobnorm{\nabla_{\PIV}\widetilde{\Loss}(\PIU,\PIV;\PIS)}^2 \leq c_3\left[\mu r \sigma_1^*\frobnorm{\PIU \PIV^{\top} - \UniqueU \UniqueV^{\top}}^2 + \mu r\sigma_1^*\delta(\delta + \epsilon\sigma_1^*)\right].
	\end{equation}
\end{lemma}

In the remainder of this section, we condition on the events in Lemma \ref{lem:local_descent_L_partial} and \ref{lem:smoothness_partial}. Now we are ready to prove Theorem \ref{thm:convergence_partial}.
\begin{proof}[Proof of Theorem \ref{thm:convergence_partial}]
	We essentially follow the process for proving Theorem \ref{thm:convergence}. Let the following shorthands be defined in the same fashion: $\delta_t$, $(\UniqueU^{t}, \UniqueV^{t})$, $(\Delta_{\PIU}^t, \Delta_{\PIV}^t)$, $\widetilde{\Loss}_t$, $\widetilde{\mathcal{G}}_t$.
	
	Here we show error decays in one step of iteration. The induction process is the same as the proof of Theorem \ref{thm:convergence}, and is thus omitted. For any $t \geq 0$, similar to \eqref{eq:delta_relation} we have that
	\begin{align*}
	\delta_{t+1} 
	& \leq \delta_t - 2\eta\underbrace{\trinprod{\nabla_{\PIU}\widetilde{\Loss}_t + \nabla_{\PIU}\widetilde{\mathcal{G}}_t}{\PIU_t - \UniqueU^t}}_{W_1} - 2\eta\underbrace{\trinprod{\nabla_{\PIV}\widetilde{\Loss}_t + \nabla_{\PIV}\widetilde{\mathcal{G}}_t}{\PIV_t - \UniqueV^t}}_{W_2} \\
	& ~~~~~~~ + \eta^2 \underbrace{\frobnorm{\nabla_{\PIU}\widetilde{\Loss}_t + \nabla_{\PIU}\widetilde{\mathcal{G}}_t}^2}_{W_3} + \eta^2\underbrace{\frobnorm{\nabla_{\PIV}\widetilde{\Loss}_t + \nabla_{\PIV}\widetilde{\mathcal{G}}_t}^2}_{W_4}.
	\end{align*}
	We also have
	\[
	\trinprod{\nabla_{\PIU}\widetilde{\Loss}_t}{\PIU_t - \UniqueU^t} + \trinprod{\nabla_{\PIV}\widetilde{\Loss}_t}{\PIV_t - \UniqueV^t} = \trinprod{\nabla_{M} \widetilde{\Loss}_t}{\PIU_t \PIV_t^{\top} - \UniqueU^t \UniqueV^{t\top} + \DeltaU^t \DeltaV^{t\top}},
	\]
	which can be lower bounded by Lemma \ref{lem:local_descent_L_partial}. Note that $\widetilde{\mathcal{G}}$ differs from $\mathcal{G}$ by a constant, we can still leverage Lemma \ref{lem:local_descent_g}. Hence, we obtain that
	\[
	W_1 + W_2 \geq \frac{1}{8}\frobnorm{\PIU_t \PIV_t^{\top} - \MStar}^2 +  \frac{1}{64}\frobnorm{\PIU_t^{\top}\PIU_t - \PIV_t^{\top}\PIV_t}^2 + \frac{1}{64}\sigma_r^*\delta_t -  \nu\sigma_1^*\delta_t - 11\sqrt{\sigma_1^*\delta_t^3} - 2\delta_t^2.
	\] 
	
	On the other hand, we have
	\begin{align*}
	W_3 + W_4 & \leq 2\frobnorm{\nabla_{\PIU}\widetilde{\Loss}_t}^2 + 2\frobnorm{\nabla_{\PIU}\widetilde{\mathcal{G}}_t}^2 +  2\frobnorm{\nabla_{\PIV}\widetilde{\Loss}_t}^2 + 2\frobnorm{\nabla_{\PIV}\widetilde{\mathcal{G}}_t}^2 \\
	& \leq c\left[\mu r \sigma_1^*\frobnorm{\PIU_t \PIV_t^{\top} - \MStar}^2 + \mu r\sigma_1^*\delta_t(\delta_t + \epsilon\sigma_1^*) + \sigma_1^*\frobnorm{\PIU_t^{\top}\PIU_t - \PIV_t^{\top}\PIV_t}^2\right],
	\end{align*}
	where $c$ is a constant, and the last step is implied by Lemma \ref{lem:smoothness} and Lemma \ref{lem:smoothness_partial}.
	
	By the assumption $\eta = c'/[\mu r\sigma_1^*]$ for sufficiently small constant $c'$, we thus have
	\begin{align*}
	-2\eta(W_1 + W_2) + \eta^2(W_3 + W_4) \leq  - \frac{1}{32}\eta\sigma_r^*\delta_t + 2\eta  \nu\sigma_1^*\delta_t +  22\eta\sqrt{\sigma_1^*\delta_t^3} + 4\eta\delta_t^2.
	\end{align*}
	Recall that $\nu := (14\beta + 81)\alpha \mu r + 26\sqrt{\epsilon} + 18\beta^{-1}$. By letting $\beta = c_1\kappa$, $\epsilon = c_2/\kappa^2$ and assuming $\alpha \leq c_3/(\mu r \kappa^2)$ and $\delta_t \leq c_4\sigma_r^*/\kappa$ for some sufficiently small constants $\{c_i\}_{i=1}^4$, we can have $-2\eta(W_1 + W_2) + \eta^2(W_3 + W_4) \leq -\frac{1}{64}\eta\sigma_r^*\delta_t$, which implies that
	\begin{equation*} 
	\delta_{t+1} \leq \left(1 - \frac{\eta \sigma_r^*}{64}\right)\delta_t,
	\end{equation*}	
	and thus completes the proof.
\end{proof}

\subsection{Proof of Corollary \ref{cor:rpca_partial}} \label{proof:cor:rpca_partial}
We need $\alpha \lesssim \frac{1}{\mu \kappa^2 r}$ due to the condition of Theorem \ref{thm:convergence_partial}. Letting the initial error provided in Theorem \ref{thm:initial_partial} be less than the corresponding condition in Theorem \ref{thm:convergence_partial}, we have
\[
51\sqrt{\kappa} \alpha \mu r\sqrt{r} \sqrt{\sigma_1^*} + 7c_1\epsilon \sqrt{\kappa \sigma_1^*} \lesssim \sqrt{\sigma_r^*/\kappa},
\]
which leads to
\[
\alpha \lesssim \frac{1}{\mu \sqrt{r^3\kappa^3}}, ~~ \epsilon \lesssim \frac{1}{\sqrt{\kappa^3}}.
\]
Plugging the above two upper bounds into the second term in \eqref{eq:condition_thm_initial_partial}, it suffices to have
\[
p \gtrsim \frac{\kappa^3\mu r^2\log d}{d_1 \wedge d_2}.
\]
Comparing the above bound with the second term in \eqref{eq:condition_thm_convergence_partial} completes the proof.	

\subsection{Proof of Lemma \ref{lem:local_descent_L}} \label{proof:lem:local_descent_L}
Let $M := \PIU\PIV^{\top}$. We observe that
\[
\nabla_{M} \mathcal{L}(\PIU, \PIV; S) = M + S - \MStar - \SStar.
\]
Plugging it back into the left hand side of \eqref{eq:local_descent}, we obtain
\begin{align}
&\trinprod{\nabla_{M} \mathcal{L}(\PIU, \PIV; S)}{ \PIU \PIV^{\top} - \UniqueU\UniqueV^{\top} + \DeltaU \DeltaV^{\top} } = \trinprod{M + S - \MStar - \SStar}{ M - \MStar + \DeltaU \DeltaV^{\top} }\notag \\
& \geq \frobnorm{M - \MStar}^2 - \underbrace{\abs{   \trinprod{S - \SStar}{M - \MStar}}}_{T_1} - \underbrace{\abs{\trinprod{ M + S - \MStar - \SStar}{\DeltaU \DeltaV^{\top}}}}_{T_2}. \label{eq:tmp_local_descent}
\end{align}

Next we derive upper bounds of $T_1$ and $T_2$ respectively.

\paragraph{Upper bound of $T_1$.} We denote the support of $S$, $\SStar$ by $\Omega$ and $\Omega^*$ respectively. Since $S - \SStar$ is supported on $\Omega^*\cup \Omega$, we have
\[
T_1 \leq \underbrace{\abs{\trinprod{\Proj_{\Omega}(S - \SStar)}{M - \MStar}}}_{W_1} + \underbrace{\abs{\trinprod{\Proj_{\Omega^* \setminus \Omega}(S - \SStar)}{M - \MStar}}}_{W_2}.
\] 
Recall that for any $(i,j) \in \Omega$, we have $\entry{S}{i}{j} = \entry{(\MStar + \SStar - M)}{i}{j}$. Accordingly, we have
\begin{equation} \label{eq:T1}
W_1 = \frobnorm{\Proj_{\Omega}(M - \MStar)}^2.
\end{equation}

Now we turn to bound $W_2$. Since $S_{(i,j)} = 0$ for any $(i,j) \in \Omega^*\setminus \Omega$, we have
\[
W_2 = \abs{\trinprod{ \Proj_{\Omega^*\setminus \Omega}\SStar}{M - \MStar}}.
\]
Let $u_i$ be the $i$-th row of $M - \MStar$, and $v_j$ be the $j$-th column of $M - \MStar$. For any $k \in [d_2]$, we let $u_i^{(k)}$ denote the element of $u_i$ that has the $k$-th largest magnitude. Similarly, for any $k \in [d_1]$, we let $v_j^{(k)}$ denote the element of $v_j$ that has the $k$-th largest magnitude.

From the design of sparse estimator \eqref{eq:sparse_estimator}, we have that for any $(i,j) \in \Omega^*\setminus \Omega$, $\abs{\entry{(\MStar + \SStar - M)}{i}{j}}$ is either smaller than the $\gamma \alpha d_2$-th largest entry of the $i$-th row of $\MStar + \SStar - M$ or smaller than the $\gamma\alpha d_1$-th largest entry of the $j$-th column of $\MStar + \SStar - M$. Note that $\SStar$ only contains at most $\alpha$-fraction nonzero entries per row and column. As a result, $\abs{\entry{(\MStar + \SStar - M)}{i}{j}}$ has to be less than the magnitude of $u_i^{(\gamma \alpha d_2 - \alpha d_2)}$ or $v_j^{(\gamma \alpha d_1 - \alpha d_1)}$. Formally, we have for $(i,j) \in \Omega^*\setminus \Omega$,
\begin{equation} \label{eq:def_b_ij}
\abs{\entry{(\MStar + \SStar - M)}{i}{j}} \leq \underbrace{\max\left\{ \abs{u_i^{(\gamma \alpha d_2 - \alpha d_2)}},~ \abs{v_j^{(\gamma \alpha d_1 - \alpha d_1)}}\right\}}_{b_{ij}}.
\end{equation}
Furthermore, we obtain
\begin{equation} \label{eq:b_ij}
b_{ij}^2 \leq  \abs{u_i^{(\gamma \alpha d_2 - \alpha d_2)}}^2 + \abs{v_j^{(\gamma \alpha d_1 - \alpha d_1)}}^2 \leq \frac{\twonorm{u_i}^2}{(\gamma - 1)\alpha d_2} + \frac{\twonorm{v_j}^2}{(\gamma - 1)\alpha d_1}.
\end{equation}
Meanwhile, for any $(i,j) \in \Omega^*\setminus \Omega$, we have
\begin{align}
\abs{\entry{\SStar}{i}{j} \cdot \entry{(M - \MStar)}{i}{j}} & = \abs{  \entry{(\MStar  + \SStar - M - \MStar + M)}{i}{j} \cdot \entry{(M - \MStar)}{i}{j} } \notag \\
& \leq \abs{\entry{(M - \MStar)}{i}{j}}^2 + \abs{\entry{(\MStar + \SStar - M)}{i}{j}}\cdot\abs{\entry{(M - \MStar)}{i}{j}} \notag\\
& \leq  \abs{\entry{(M - \MStar)}{i}{j}}^2 + b_{ij}\cdot\abs{(M - \MStar)_{(i,j)}} \notag \\
& \leq \left(1 + \frac{\beta}{2}\right)\abs{\entry{(M - \MStar)}{i}{j}}^2 + \frac{b_{ij}^2}{2\beta}, \label{eq:inner_product}
\end{align}
where $\beta$ in the last step can be any positive number. Combining \eqref{eq:b_ij} and \eqref{eq:inner_product} leads to
\begin{align}
W_2 & \leq \sum_{(i,j) \in \Omega^*\setminus \Omega}  \abs{\entry{\SStar}{i}{j} \cdot \entry{(M - \MStar)}{i}{j}} \notag\\
& \leq \left( 1 + \frac{\beta}{2}\right)\frobnorm{\Proj_{\Omega^*\setminus \Omega}(M - \MStar)}^2 + \sum_{(i, j) \in \Omega^* \setminus \Omega} \frac{b_{ij}^2}{2\beta} \notag\\
& \leq  \left( 1 + \frac{\beta}{2}\right)\frobnorm{\Proj_{\Omega^*\setminus \Omega}(M - \MStar)}^2  + \frac{1}{2\beta} \sum_{(i, j) \in \Omega^* \setminus \Omega} \left( \frac{\twonorm{u_i}^2}{(\gamma - 1)\alpha d_2} + \frac{\twonorm{v_j}^2}{(\gamma - 1)\alpha d_1} \right) \notag\\
& \leq  \left( 1 + \frac{\beta}{2}\right)\frobnorm{\Proj_{\Omega^*\setminus \Omega}(M - \MStar)}^2 + \frac{1}{\beta(\gamma - 1)}\frobnorm{M - \MStar}^2. \label{eq:T2}
\end{align}
In the last step, we use
\begin{align} \label{eq:tmp3}
& \sum_{(i,j) \in \Omega^*\setminus \Omega} \left(\frac{1}{d_2}\twonorm{u_i}^2 + \frac{1}{d_1}\twonorm{v_j}^2 \right) \leq \sum_{(i,j) \in \Omega^*} \left(\frac{1}{d_2}\twonorm{u_i}^2 + \frac{1}{d_1}\twonorm{v_j}^2 \right) \notag \\
& \leq \sum_{i \in [d]}\sum_{j \in \row{\Omega^*}{i}} \frac{1}{d_2}\twonorm{u_i}^2 + \sum_{j \in [d]}\sum_{i \in \column{\Omega^*}{j}} \frac{1}{d_1}\twonorm{v_j}^2 \notag\\
& \leq \alpha \sum_{i \in [d]}\twonorm{u_i}^2 + \alpha \sum_{j \in [d]}\twonorm{v_j}^2 = 2\alpha\frobnorm{M - \MStar}^2.
\end{align}
We introduce shorthand $\delta := \frobnorm{\DeltaU}^2 + \frobnorm{\DeltaV}^2$. We prove the following inequality in the end of this section.
\begin{equation} \label{eq:M_frobnorm_bound}
\frobnorm{M - \MStar} \leq \sqrt{5\sigma_1^*\delta}.
\end{equation}

Combining \eqref{eq:T1}, \eqref{eq:T2} and \eqref{eq:M_frobnorm_bound} leads to
\begin{align}
T_1 & \leq \frobnorm{\Proj_{\Omega}(M - \MStar)}^2 + \left(1 + \frac{\beta}{2}\right)\frobnorm{\Proj_{\Omega^*\setminus \Omega}(M - \MStar)}^2 + \frac{5\sigma_1^*\delta}{\beta(\gamma  - 1)} \notag\\
& \leq 9(2\gamma + \beta + 2)\alpha \mu r \sigma_1^*\delta + \frac{5\sigma_1^*\delta}{\beta(\gamma  - 1)}, \label{eq:V1}
\end{align}
where the last step follows from Lemma \ref{lem:projection} by noticing that $\Proj_{\Omega}(M - \MStar)$ has at most $\gamma\alpha$-fraction nonzero entries per row and column.

\paragraph{Upper bound of $T_2$.} To ease notation, we let $C := M + S - \MStar - \SStar$. We observe that $C$ is supported on $\Omega^c$, we have
\[
T_2 \leq \underbrace{\abs{\trinprod{\Proj_{\Omega^{*c} \cap \Omega^c} (M - \MStar)}{ \DeltaU \DeltaV^{\top} }}}_{W_3} + \underbrace{\abs{\trinprod{ \Proj_{\Omega^{*} \cap \Omega^c}C}{\DeltaU \DeltaV^{\top}}}}_{W_4}.
\]
By Cauchy-Swartz inequality, we have 
\[
W_3 \leq \frobnorm{\Proj_{\Omega^{*c} \cap \Omega^c} (M - \MStar)}\frobnorm{\DeltaU \DeltaV^{\top} } \leq \frobnorm{M - \MStar}\frobnorm{\DeltaU}\frobnorm{\DeltaV} \leq \sqrt{5\sigma_1^*\delta^3}/2,
\]
where the last step follows from \eqref{eq:M_frobnorm_bound} and  $\frobnorm{\DeltaU}\frobnorm{\DeltaV} \leq \delta/2$.

It remains to bound $W_4$. By Cauchy-Swartz inequality, we have
\begin{align*}
W_4 & \leq \frobnorm{\Proj_{\Omega^*\cap \Omega^c}C}\frobnorm{\DeltaU \DeltaV^{\top}} \leq \frobnorm{\Proj_{\Omega^*\cap \Omega^c}(\MStar + \SStar - M)}\frobnorm{\DeltaU \DeltaV^{\top}} \\
& \overset{(a)}{\leq} \sqrt{\sum_{(i,j)\in \Omega^*\setminus \Omega} b_{ij}^2} \frobnorm{\DeltaU}\frobnorm{\DeltaV} \overset{(b)}{\leq} \left[\sum_{(i,j) \in \Omega^*\setminus \Omega} \frac{\twonorm{u_i}^2}{(\gamma - 1)\alpha d_2} + \frac{\twonorm{v_j}^2}{(\gamma - 1)\alpha d_1} \right]^{1/2}\frobnorm{\DeltaU}\frobnorm{\DeltaV}. \\
& \overset{(c)}{\leq} \sqrt{\frac{2}{\gamma - 1}}\frobnorm{M - \MStar}\frobnorm{\DeltaU}\frobnorm{\DeltaV} \leq  \sqrt{\frac{5\sigma_1^*\delta^3}{2(\gamma - 1)}},
\end{align*}
where step $(a)$ is from \eqref{eq:def_b_ij}, step $(b)$ follows from \eqref{eq:b_ij}, and step $(c)$ follows from \eqref{eq:tmp3}. Combining the upper bounds of $W_3$ and $W_4$, we obtain
\begin{equation} \label{eq:V2}
T_2 \leq  \sqrt{5\sigma_1^*\delta^3}/2 + \sqrt{\frac{5\sigma_1^*\delta^3}{2(\gamma - 1)}}.
\end{equation}

\paragraph{Combining pieces.} Now we choose $\gamma = 2$. Then inequality \eqref{eq:V1} implies that
\[
T_1 \leq [9(\beta+6)\alpha \mu r + 5\beta^{-1}]\sigma_1^*\delta.
\]
Inequality \eqref{eq:V2} then implies that
\[
T_2 \leq 3\sqrt{\sigma_1^*\delta^3}.
\]

Plugging the above two inequalities into \eqref{eq:tmp_local_descent} completes the proof.

\vskip .1in
\noindent{\em Proof of inequality \eqref{eq:M_frobnorm_bound}.}
We find that
\begin{align*}
\frobnorm{M - \MStar}^2 & \leq \left[\sqrt{\sigma_1^*}(\frobnorm{\DeltaV} + \frobnorm{\DeltaU}) + \frobnorm{\DeltaU}\frobnorm{\DeltaV}\right]^2 \\
& \leq \left[\sqrt{\sigma_1^*}(\frobnorm{\DeltaV} + \frobnorm{\DeltaU}) + \frac{1}{2}\sqrt{\sigma_1^*}\frobnorm{\DeltaU} + \frac{1}{2}\sqrt{\sigma_1^*}\frobnorm{\DeltaV}\right]^2 \\
& \leq 5\sigma_1^*(\frobnorm{\DeltaU}^2 + \frobnorm{\DeltaV}^2),
\end{align*}
where the first step follows from the upper bound of $\frobnorm{M - \MStar}$ shown in Lemma \ref{lem:delta_norm}, and the second step follows from the assumption $\frobnorm{\DeltaU}, \frobnorm{\DeltaV} \leq \sqrt{\sigma_1^*}$.

\subsection{Proof of Lemma \ref{lem:local_descent_g}} \label{proof:lem:local_descent_g}
We first observe that
\[
\nabla_{\PIU} \mathcal{G}(\PIU, \PIV) = \frac{1}{2}\PIU(\PIU^{\top}\PIU - \PIV^{\top}\PIV),  ~\nabla_{\PIV} \mathcal{G}(\PIU, \PIV) = \frac{1}{2}\PIV(\PIV^{\top}\PIV - \PIU^{\top}\PIU),
\]
Therefore, we obtain
\begin{align}
& \trinprod{\nabla_{\PIU} \mathcal{G}(\PIU, \PIV)}{U -  \UniqueU} + \trinprod{\nabla_{\PIV} \mathcal{G}(\PIU, \PIV)}{V -  \UniqueV}  \notag\\
& = \frac{1}{2}\trinprod{\PIU^{\top}\PIU - \PIV^{\top}\PIV}{\PIU^{\top}\PIU - \PIV^{\top}\PIV - \PIU^{\top}\UniqueU + \PIV^{\top}\UniqueV} \notag\\
& = \frac{1}{4}\frobnorm{\PIU^{\top}\PIU - \PIV^{\top}\PIV}^2 +  \frac{1}{4}\trinprod{\PIU^{\top}\PIU - \PIV^{\top}\PIV}{\PIU^{\top}\PIU - \PIV^{\top}\PIV - 2\PIU^{\top}\UniqueU + 2\PIV^{\top}\UniqueV} \notag\\
& =  \frac{1}{4}\frobnorm{\PIU^{\top}\PIU - \PIV^{\top}\PIV}^2 +  \frac{1}{4}\trinprod{\PIU^{\top}\PIU - \PIV^{\top}\PIV}{\PIU^{\top}\PIU - \PIV^{\top}\PIV - 2\DeltaU^{\top}\UniqueU + 2\DeltaV^{\top}\UniqueV}, \label{eq:lower_G}
\end{align}
where the last step follows from $\DeltaU^{\top}\UniqueU - \DeltaV^{\top}\UniqueV = \PIU^{\top}\UniqueU - \PIV^{\top}\UniqueV$ since $\UniqueU^{\top}\UniqueU = \UniqueV^{\top}\UniqueV$.
Note that 
\begin{align*}
\PIU^{\top}\PIU - \PIV^{\top}\PIV & = (\UniqueU + \DeltaU)^{\top}(\UniqueU + \DeltaU) -  (\UniqueV + \DeltaV)^{\top}(\UniqueV + \DeltaV)\\
& = \UniqueU^{\top}\DeltaU + \DeltaU^{\top}\UniqueU + \DeltaU^{\top}\DeltaU - \UniqueV^{\top}\DeltaV - \DeltaV^{\top}\UniqueV - \DeltaV^{\top}\DeltaV,
\end{align*} 
where we use $\UniqueU^{\top}\UniqueU = \UniqueV^{\top}\UniqueV$ again in the last step. Furthermore, since $\PIU^{\top}\PIU - \PIV^{\top}\PIV$ is symmetric, we have 
\begin{align*}
& \trinprod{\PIU^{\top}\PIU - \PIV^{\top}\PIV}{\UniqueU^{\top}\DeltaU + \DeltaU^{\top}\UniqueU - \UniqueV^{\top}\DeltaV - \DeltaV^{\top}\UniqueV} \\
& = \trinprod{\PIU^{\top}\PIU - \PIV^{\top}\PIV}{2\DeltaU^{\top}\UniqueU - 2\DeltaV^{\top}\UniqueV}.
\end{align*}
Using these arguments, for the second term in \eqref{eq:lower_G}, denoted by $T_2$, we have
\[
T_2 = \frac{1}{4}\trinprod{\PIU^{\top}\PIU - \PIV^{\top}\PIV}{\DeltaU^{\top}\DeltaU - \DeltaV^{\top}\DeltaV}.
\]
Furthermore, we have
\begin{align}
4T_2 & \leq \abs{\trinprod{\PIU^{\top}\PIU - \PIV^{\top}\PIV}{\DeltaU^{\top}\DeltaU - \DeltaV^{\top}\DeltaV}} \leq \frobnorm{\PIU^{\top}\PIU - \PIV^{\top}\PIV}\left(\frobnorm{\DeltaU}^2 + \frobnorm{\DeltaV}^2\right) \notag\\
& \leq \left( \frobnorm{\PIU^{\top}\PIU  - \UniqueU^{\top}\UniqueU} + \frobnorm{\PIV^{\top}\PIV  - \UniqueV^{\top}\UniqueV}\right) \delta \notag\\
& \leq 2\left(\opnorm{\UniqueU}\frobnorm{\DeltaU} + \opnorm{\UniqueV}\frobnorm{\DeltaV}\right)\delta \leq 2\sqrt{2\sigma_1^*\delta^3}. \label{eq:GT2}
\end{align}

It remains to find a lower bound of $\frobnorm{\PIU^{\top}\PIU - \PIV^{\top}\PIV}$. The following inequality, which we turn to prove later, is true:
\begin{equation} \label{eq:it}
\frobnorm{\PIU^{\top}\PIU - \PIV^{\top}\PIV}^2 \geq \frobnorm{\PIU \PIU^{\top} - \UniqueU \UniqueU^{\top}}^2 + \frobnorm{\PIV \PIV^{\top} - \UniqueV \UniqueV^{\top}}^2 - 2\frobnorm{\PIU \PIV^{\top} - \UniqueU \UniqueV^{\top}}^2.
\end{equation}
Proceeding with the first term in \eqref{eq:lower_G} by using \eqref{eq:it}, we get
\begin{align}
& \frac{1}{4}\frobnorm{\PIU^{\top}\PIU - \PIV^{\top}\PIV}^2 = \frac{1}{8}\frobnorm{\PIU^{\top}\PIU - \PIV^{\top}\PIV}^2 + \frac{1}{8}\frobnorm{\PIU^{\top}\PIU - \PIV^{\top}\PIV}^2 \notag\\
& \geq \frac{1}{8}\frobnorm{\PIU^{\top}\PIU - \PIV^{\top}\PIV}^2 + \frac{1}{8}\frobnorm{\PIU \PIU^{\top} - \UniqueU \UniqueU^{\top}}^2 + \frac{1}{8}\frobnorm{\PIV \PIV^{\top} - \UniqueV \UniqueV^{\top}}^2 - \frac{1}{4}\frobnorm{\PIU \PIV^{\top} - \UniqueU \UniqueV^{\top}}^2 \notag\\
& = \frac{1}{8}\frobnorm{\PIU^{\top}\PIU - \PIV^{\top}\PIV}^2 + \frac{1}{8}\frobnorm{FF^{\top} - F_{\pi^*}F^{\top}_{\pi^*}}^2 - \frac{1}{2}\frobnorm{\PIU \PIV^{\top} - \UniqueU\UniqueV^{\top}}^2, \label{eq:lower_bound_G}
\end{align}
where we let
\[
F := \begin{bmatrix}
\PIU \\ \PIV\end{bmatrix}, ~~  F_{\pi^*} := \begin{bmatrix}
\UniqueU \\ \UniqueV\end{bmatrix}.
\]

Introduce $\Delta_F := F - F_{\pi^*}$. Recall that $\delta := \frobnorm{\DeltaU}^2 + \frobnorm{\DeltaV}^2$. Equivalently $\delta = \frobnorm{\Delta_F}^2$. We have
\begin{align*}
\frobnorm{FF^{\top} - F_{\pi^*}F^{\top}_{\pi^*}} & = \frobnorm{\Delta_F F^{\top}_{\pi*} + F_{\pi^*}\Delta_F^{\top} + \Delta_F\Delta_F^{\top}} \\
& \geq \frobnorm{\Delta_F F^{\top}_{\pi*} + F_{\pi^*}\Delta_F^{\top}} - \frobnorm{\Delta_F}^2 = \frobnorm{\Delta_F F^{\top}_{\pi*} + F_{\pi^*}\Delta_F^{\top}} - \delta.
\end{align*}
For the first term, we have
\begin{align*}
& \frobnorm{\Delta_F F^{\top}_{\pi*} + F_{\pi^*}\Delta_F^{\top}}^2  = 2\frobnorm{\Delta_F F^{\top}_{\pi*}}^2 + \trinprod{\Delta_F F^{\top}_{\pi*}}{F_{\pi^*}\Delta_F^{\top}} \\
& \geq  2\sigma_r(F_{\pi^*})^2 \frobnorm{\Delta_F}^2 + \trinprod{\Delta_F F^{\top}_{\pi*}}{F_{\pi^*}\Delta_F^{\top}} = 4\sigma_r^*\frobnorm{\Delta_F}^2 + 	\trinprod{\Delta_F F^{\top}_{\pi*}}{F_{\pi^*}\Delta_F^{\top}}.
\end{align*}
For the cross term, by the following result, proved in \cite{chen2015fast} (we also provide a proof in Section \ref{proof:lem:symm} for the sake of completeness), we have $\trinprod{\Delta_F F^{\top}_{\pi*}}{F_{\pi^*}\Delta_F^{\top}} \geq 0$.
\begin{lemma} \label{lem:symm}
	When $\opnorm{F - F_{\pi^*}} < \sqrt{2\sigma_r^*}$, we have that $\Delta_F^{\top}F_{\pi^*}$ is symmetric.
\end{lemma}
Accordingly, we have $\frobnorm{FF^{\top} - F_{\pi^*}F^{\top}_{\pi^*}} \geq 2\sqrt{\sigma_r^*\delta}- \delta \geq \sqrt{\sigma_r^*\delta}$ under condition $\delta \leq \sigma_r^*$. Plugging this lower bound into \eqref{eq:lower_bound_G}, we obtain
\[
\frac{1}{4}\frobnorm{\PIU^{\top}\PIU - \PIV^{\top}\PIV}^2 \geq \frac{1}{8}\frobnorm{\PIU^{\top}\PIU - \PIV^{\top}\PIV}^2 + \frac{1}{8}\sigma_r^*\delta - \frac{1}{2}\frobnorm{\PIU \PIV^{\top} - \UniqueU\UniqueV^{\top}}^2.
\]
Putting \eqref{eq:lower_G}, \eqref{eq:GT2} and the above inequality together completes the proof.

\vskip 0.1in
\noindent{\em Proof of inequality \eqref{eq:it}.} For the term on the left hand side of \eqref{eq:it}, it is easy to check that
\begin{equation} \label{eq:tmp_local_G}
\frobnorm{\PIU^{\top}\PIU - \PIV^{\top}\PIV}^2 = \frobnorm{\PIU\PIU^{\top}}^2 + \frobnorm{\PIV\PIV^{\top}}^2 - 2\frobnorm{\PIU\PIV^{\top}}^2.
\end{equation}
The property $\UniqueU^{\top} \UniqueU = \UniqueV^{\top} \UniqueV$ implies that $\frobnorm{\UniqueU\UniqueU^{\top}} = \frobnorm{\UniqueV\UniqueV^{\top}} = \frobnorm{\UniqueU\UniqueV^{\top}}$. Therefore, expanding those quadratic terms on the right hand side of \eqref{eq:it}, one can show that it is equal to 
\begin{equation} \label{eq:tmp_local_G_1}
\frobnorm{\PIU\PIU^{\top}}^2 + \frobnorm{\PIV\PIV^{\top}}^2 - 2\frobnorm{\UniqueU^{\top}\PIU}^2 - 2\frobnorm{\UniqueV^{\top}\PIV}^2 + 4\trinprod{\UniqueU^{\top}\PIU}{\UniqueV^{\top}\PIV} - 2\frobnorm{\PIU\PIV^{\top}}^2.
\end{equation}
Comparing inequalities \eqref{eq:tmp_local_G} and \eqref{eq:tmp_local_G_1}, it thus remains to show that 
\[
-2\frobnorm{\UniqueU^{\top}\PIU}^2 - 2\frobnorm{\UniqueV^{\top}\PIV}^2 + 4\trinprod{\UniqueU^{\top}\PIU}{\UniqueV^{\top}\PIV} \leq 0.
\]
Equivalently, we always have $\frobnorm{\UniqueU^{\top}\PIU - \UniqueV^{\top}\PIV}^2 \geq 0$, and thus prove \eqref{eq:it}.

\subsection{Proof of Lemma \ref{lem:smoothness}} \label{proof:lem:smoothness}
First, we turn to prove \eqref{eq:smooth_G}. As
\[
\nabla_{\PIU} \mathcal{G}(\PIU, \PIV) = \frac{1}{2}\PIU(\PIU^{\top}\PIU - \PIV^{\top}\PIV),  ~\nabla_{\PIV} \mathcal{G}(\PIU, \PIV) = \frac{1}{2}\PIV(\PIV^{\top}\PIV - \PIU^{\top}\PIU),
\]
we have
\[
\frobnorm{\nabla_{\PIU} \mathcal{G}(\PIU, \PIV)}^2 + \frobnorm{\nabla_{\PIV} \mathcal{G}(\PIU, \PIV)}^2 \leq \frac{1}{4}\left(\opnorm{U}^2 + \opnorm{V}^2\right)\frobnorm{\PIU^{\top}\PIU - \PIV^{\top}\PIV}^2.
\]
As $(\PIU, \PIV) \in \mathbb{B}_2(\sqrt{\sigma_1^*})$, we thus have $\opnorm{\PIU} \leq \opnorm{\UniqueU} + \opnorm{\UniqueU - \PIU} \leq 2\sqrt{\sigma_1^*}$, and similarly $\opnorm{\PIV} \leq 2\sqrt{\sigma_1^*}$. We obtain
\[
\frobnorm{\nabla_{\PIU} \mathcal{G}(\PIU, \PIV)}^2 + \frobnorm{\nabla_{\PIV} \mathcal{G}(\PIU, \PIV)}^2 \leq 2\sigma_1^*\frobnorm{\PIU^{\top}\PIU - \PIV^{\top}\PIV}^2.
\] 
\vskip .1in
Now we turn to prove \eqref{eq:smooth_L}. We observe that
\[
\nabla_{M} \Loss(\PIU, \PIV; \PIS) = M + \PIS - \MStar - \SStar,
\]
where we let $M := \PIU\PIV^{\top}$.
We denote the support of $S, \SStar$ by $\Omega$ and $\Omega^*$ respectively. Based on the sparse estimator \eqref{eq:sparse_estimator} for computing $S$, $\nabla_{M} \Loss(\PIU, \PIV; \PIS)$ is only supported on $\Omega^c$. We thus have
\begin{align*}
\frobnorm{\nabla_{M} \Loss(\PIU, \PIV; \PIS)} & \leq \frobnorm{\Proj_{\Omega^c\setminus \Omega^*}(M - \MStar)} + \frobnorm{\Proj_{\Omega^c \cap \Omega^*}(M - \MStar - \SStar)} \\
& \leq \frobnorm{M - \MStar} + \frobnorm{\Proj_{\Omega^c \cap \Omega^*}(M - \MStar - \SStar)}.
\end{align*}
It remains to upper bound the second term on the right hand side. Following \eqref{eq:def_b_ij} and \eqref{eq:b_ij}, we have
\begin{align*}
& \frobnorm{\Proj_{\Omega^c \cap \Omega^*}(M - \MStar - \SStar)}^2  \leq \sum_{(i,j)\in \Omega^c \cap \Omega^*} \frac{\twonorm{u_i}^2}{(\gamma - 1)\alpha d_2} + \frac{\twonorm{v_j}^2}{(\gamma - 1)\alpha d_1} \leq \frac{2}{\gamma - 1}\frobnorm{M - \MStar}^2,
\end{align*}
where the last step is proved in \eqref{eq:tmp3}. By choosing $\gamma = 2$,
we thus conclude that
\[
\frobnorm{\nabla_{M} \Loss(\PIU, \PIV; \PIS)} \leq (1+\sqrt{2})\frobnorm{M - \MStar}.
\]

\subsection{Proof of Lemma \ref{lem:local_descent_L_partial}} \label{proof:lem:local_descent_L_partial}
We denote the support of $\Proj_{\ObserveSet}(\SStar)$, $S$ by $\ObserveSparse$ and $\Omega$. We always have $\ObserveSparse \subseteq \ObserveSet$ and $\Omega \subseteq \ObserveSet$.

In the sequel, we establish several results that characterize the properties of $\ObserveSet$. The first result, proved in Section \ref{proof:lem:norm_sampling}, shows that the Frobenius norm of any incoherent matrix whose row (or column) space are equal to $\PIL^*$ (or $\PIR^*$) is well preserved under partial observations supported on $\ObserveSet$.
\begin{lemma} \label{lem:norm_sampling}
	Suppose $\MStar \in \real^{d_1 \times d_2}$ is a rank $r$ and $\mu$-incoherent matrix that has SVD $\MStar = \PIL^*\SVMat^* \PIR^{*\top}$. Then there exists an absolute constant $c$ such that for any $\epsilon \in (0,1)$, if $p \geq c\frac{\mu r\log d}{\epsilon^2 (d_1 \wedge d_2)}$, then with probability at least $1 - 2d^{-3}$, we have that for all $A \in \real^{d_2 \times r}, B \in \real^{d_1\times r}$,
	\[ 
	(1-\epsilon)\frobnorm{\PIL^*A^{\top} + B\PIR^{*\top}}^2 \leq p^{-1}\frobnorm{\Proj_{\ObserveSet}\left(\PIL^*A^{\top} + B\PIR^{*\top}\right)}^2 \leq (1+\epsilon)\frobnorm{\PIL^*A^{\top} + B\PIR^{*\top}}^2.
	\]
\end{lemma}

We need the next result, proved in Section \ref{proof:lem:size_concentration}, to control the number of nonzero entries per row and column in $\ObserveSparse$ and $\ObserveSet$.
\begin{lemma} \label{lem:size_concentration}
	If $p \geq \frac{56}{3}\frac{\log d}{\alpha (d_1 \wedge d_2)}$, then with probability at least $1 - 6d^{-1}$, we have
	\begin{equation*}
	\big| \abs{\row{\ObserveSet}{i}} - pd_2\big| \leq \frac{1}{2}pd_2,~~ \big|\abs{\column{\ObserveSet}{j}} - pd_1\big| \leq \frac{1}{2}pd_1, ~~ \abs{\Omega^*_{o(i,\cdot)}} \leq \frac{3}{2}\alpha p d_2, ~~ \abs{\Omega^*_{o(\cdot, j)}} \leq \frac{3}{2}\alpha p d_1,
	\end{equation*}
	for all $i \in [d_1]$ and $j \in [d_2]$.
\end{lemma}

The next lemma, proved in Section \ref{proof:lem:frobnorm_sampling}, can be used to control the projection of small matrices to $\ObserveSet$.
\begin{lemma} \label{lem:frobnorm_sampling}
	There exists constant $c$ such that for any $\epsilon \in (0,1)$, if $p \geq c\frac{\mu^2 r^2\log d}{\epsilon^2(d_1 \wedge d_2)}$, then with probability at least $1 - \order(d^{-1})$, for all matrices $Z \in \real^{d_1 \times d_2}$, $\PIU \in \real^{d_1 \times r}$ and $\PIV \in \real^{d_2 \times r}$ that satisfy $\twoinfnorm{\PIU} \leq \sqrt{\mu r/d_1}$,$\twoinfnorm{\PIV} \leq \sqrt{\mu r/d_2}$, we have
	\begin{equation} \label{eq:uv_sampling}
	p^{-1}\frobnorm{\Proj_{\ObserveSet}(\PIU \PIV^{\top})}^2 \leq \frobnorm{\PIU}^2\frobnorm{\PIV}^2 + \epsilon \frobnorm{\PIU}\frobnorm{\PIV};
	\end{equation}
	\begin{equation} \label{eq:Zv_sampling}
	p^{-1}\frobnorm{\Proj_{\ObserveSet}(Z)V}^2 \leq 2\mu r \frobnorm{\Proj_{\ObserveSet}(Z)}^2;
	\end{equation}
	\begin{equation} \label{eq:uZ_sampling}
	p^{-1}\frobnorm{\PIU^{\top}\Proj_{\ObserveSet}(Z)}^2 \leq 2\mu r \frobnorm{\Proj_{\ObserveSet}(Z)}^2.
	\end{equation}
\end{lemma}

In the remainder of this section, we condition on the events in Lemmas \ref{lem:norm_sampling}, \ref{lem:size_concentration} and \ref{lem:frobnorm_sampling}. Now we are ready to prove Lemma \ref{lem:local_descent_L_partial}.

\begin{proof}[Proof of Lemma \ref{lem:local_descent_L_partial}]	
	Using shorthand $M := \PIU\PIV^{\top}$, we have
	\[
	\nabla_{M} \mathcal{\widetilde{L}}(\PIU, \PIV; S) = p^{-1}\Proj_{\ObserveSet}\left(M + S - \MStar - \SStar\right).
	\]
	Plugging it back into the left hand side of \eqref{eq:local_descent}, we obtain
	\begin{align}
	&\trinprod{\nabla_{M} \mathcal{\widetilde{L}}(\PIU, \PIV; S)}{ \PIU \PIV^{\top} - \UniqueU\UniqueV^{\top} + \DeltaU \DeltaV^{\top} } \notag\\
	& = \frac{1}{p}\trinprod{\Proj_{\ObserveSet}\left(M + S - \MStar - \SStar\right)}{ M - \MStar + \DeltaU \DeltaV^{\top} }\notag \\
	& \geq \underbrace{\frac{1}{p}\frobnorm{\Proj_{\ObserveSet}\left(M - \MStar\right)}^2}_{T_1} - \underbrace{\frac{1}{p}\abs{   \trinprod{\Proj_{\ObserveSet}\left(S - \SStar\right)}{M - \MStar}}}_{T_2} - \underbrace{\frac{1}{p}\abs{\trinprod{ \Proj_{\ObserveSet}\left(M + S - \MStar - \SStar\right)}{\DeltaU \DeltaV^{\top}}}}_{T_3}. \label{eq:tmp_local_descent_partial}
	\end{align}
	
	Next we derive lower bounds of $T_1$, upper bounds of $T_2$ and $T_3$ respectively.

	\paragraph{Lower bound of $T_1$.}  We observe that $M - \MStar = \UniqueU^*\DeltaV^{\top} + \DeltaU \UniqueV^{\top} + \DeltaU\DeltaV^{\top}$. By triangle inequality, we have
	\[ 
	\frobnorm{\Proj_{\ObserveSet}(M - \MStar)} \geq \frobnorm{\Proj_{\ObserveSet}(\UniqueU\DeltaV^{\top} + \DeltaU \UniqueV^{\top})} - \frobnorm{\Proj_{\ObserveSet}(\DeltaU\DeltaV^{\top})}.
	\]
	Note that when $c \geq a - b$ for $a, b \geq 0$, we always have $c^2 \geq \frac{1}{2}a^2 - b^2$. We thus have
	\begin{align*}
	T_1 & \geq \frac{1}{2p}\frobnorm{\Proj_{\ObserveSet}(\UniqueU\DeltaV^{\top} + \DeltaU \UniqueV^{\top})}^2 - \frac{1}{p}\frobnorm{\Proj_{\ObserveSet}(\DeltaU\DeltaV^{\top})}^2 \\
	& \geq \frac{1}{2}(1-\epsilon)\frobnorm{\UniqueU\DeltaV^{\top} + \DeltaU \UniqueV^{\top}}^2 - \frac{1}{p}\frobnorm{\Proj_{\ObserveSet}(\DeltaU\DeltaV^{\top})}^2 \\
	& \geq \frac{1}{2}(1-\epsilon)\frobnorm{M - \MStar - \DeltaU\DeltaV^{\top}}^2 - \frobnorm{\DeltaU}^2\frobnorm{\DeltaV}^2 - 9\epsilon\sigma_1^*\frobnorm{\DeltaU}\frobnorm{\DeltaV} \\
	& \geq \frac{1}{4}(1-\epsilon)\frobnorm{M - \MStar}^2 - \frac{1}{2}(1-\epsilon)\frobnorm{\DeltaU\DeltaV}^2 - \frobnorm{\DeltaU}^2\frobnorm{\DeltaV}^2 - 9\epsilon\sigma_1^*\frobnorm{\DeltaU}\frobnorm{\DeltaV} \\
	& \geq \frac{1}{4}(1-\epsilon)\frobnorm{M - \MStar}^2 - 2\delta^2 - 5\epsilon\sigma_1^*\delta.
	\end{align*}
	where the second step is implied by Lemma \ref{lem:norm_sampling}, the third step follows from \eqref{eq:uv_sampling} in Lemma \ref{lem:frobnorm_sampling} by noticing that $\twoinfnorm{\DeltaU} \leq 3\sqrt{\mu r \sigma_1^*/d_1}$ and  $\twoinfnorm{\DeltaV} \leq 3\sqrt{\mu r \sigma_1^*/d_1}$, which is further implied by \eqref{eq:uvcondition_partial}.

	\paragraph{\bf Upper bound of $T_2$.} 
	Since $S - \SStar$ is supported on $\Omega^*_{0}\cup \Omega$, we have
	\begin{equation} \label{eq:V2_partial}
	pT_2 \leq \abs{\trinprod{ \Proj_{\ObserveSparse\setminus \Omega}(\SStar)}{\Proj_{\ObserveSparse\setminus \Omega}(M - \MStar)}} + \abs{\trinprod{\Proj_{\Omega}(S - \SStar)}{\Proj_{\Omega}(M - \MStar)}}.
	\end{equation}
	For any $(i,j) \in \Omega$, we have $(S - \SStar)_{(i,j)} = (\MStar - M)_{(i,j)}$. Therefore, for the second term on the right hand side, we have
	\begin{equation} \label{eq:tmp4}
	\abs{\trinprod{\Proj_{\Omega}(S - \SStar)}{\Proj_{\Omega}(M - \MStar)}} \leq \frobnorm{\Proj_{\Omega}(M - \MStar)}^2 \leq 18\gamma p \alpha \mu r \sigma_1^* \delta,
	\end{equation}
	where the last inequality follows from Lemma \ref{lem:projection} and the fact that $\abs{\Omega_{(i,\cdot)}} \leq \gamma p \alpha d_2$, $\abs{\Omega_{(\cdot, j)}} \leq \gamma p \alpha d_1$ for all $i \in [d_1]$, $j \in [d_2]$.
	
	We denote the $i$-th row of $\Proj_{\ObserveSet}(M - \MStar)$ by $u_i$, and we denote the $j$-th column of $\Proj_{\ObserveSet}(M - \MStar)$ by $v_j$. We let $u_i^{(k)}$ denote the element of $u_i$ that has the $k$-th largest magnitude. We let $v_j^{(k)}$ denote the element of $v_j$ that has the $k$-th largest magnitude. 
	
	For the first term on the right hand side of \eqref{eq:V2_partial}, we first observe that for $(i,j) \in \ObserveSparse\setminus \Omega$, $\abs{(\MStar + \SStar - M)_{(i,j)}}$ is either less than the $\gamma p\alpha d_2$-th largest element in the $i$-th row of $\Proj_{\ObserveSet}(\MStar + \SStar - M)$, or less than $\gamma p \alpha d_1$-th largest element in the $j$-th row of $\Proj_{\ObserveSet}(\MStar + \SStar - M)$. Based on Lemma \ref{lem:size_concentration}, $\Proj_{\ObserveSet}(\SStar)$ has at most $3p\alpha d_2/2$ nonzero entries per row and at most $3p \alpha d_1/2$ nonzero entries per column. Therefore, we have
	\begin{equation} \label{eq:bij_partial}
	\abs{\entry{(\MStar + \SStar - M)}{i}{j}} \leq \max\left\{ \abs{u_i^{((\gamma - 1.5)p\alpha d_2)}},~ \abs{v_j^{((\gamma - 1.5)p\alpha d_1)}}\right\}.
	\end{equation}
	In addition, we observe that
	\begin{align}
	& \abs{\trinprod{ \Proj_{\ObserveSparse\setminus \Omega}(\SStar)}{ \Proj_{\ObserveSparse\setminus \Omega}(M - \MStar)}} \notag \\
	&  \leq \sum_{(i,j) \in \ObserveSparse\setminus \Omega} \abs{(\MStar + \SStar - M)_{(i,j)}}\abs{(\MStar - M)_{(i,j)}} + \abs{(\MStar - M)_{(i,j)}}^2 \notag\\
	& \leq \left(1 + \frac{\beta}{2}\right)\frobnorm{\Proj_{\ObserveSparse}(\MStar - M)}^2 + \frac{1}{2\beta}\sum_{(i,j) \in \ObserveSparse\setminus \Omega} \abs{(\MStar + \SStar - M)_{(i,j)}}^2, \notag \\
	& \leq (27 + 14\beta)p\alpha\mu r\sigma_1^*\delta + \frac{1}{2\beta}\sum_{(i,j) \in \ObserveSparse\setminus \Omega} \abs{(\MStar + \SStar - M)_{(i,j)}}^2, \label{eq:tmp5}
	\end{align}
	where the second step holds for any $\beta > 0$ and the last step follows from Lemma \ref{lem:projection} under the size constraints of $\ObserveSparse$ shown in Lemma \ref{lem:size_concentration}. For the second term in \eqref{eq:tmp5}, using \eqref{eq:bij_partial}, we have
	\begin{align} \label{eq:tmp6}
	& \sum_{(i,j) \in \ObserveSparse\setminus \Omega} \abs{(\MStar + \SStar - M)_{(i,j)}}^2 \leq \sum_{(i,j) \in \ObserveSparse} \abs{u_i^{((\gamma - 1.5)p\alpha d_2)}}^2 + \abs{v_j^{((\gamma - 1.5)p\alpha d_1)}}^2 \notag \\
	& = \sum_{i \in [d_1]} \sum_{j \in \Omega^*_{o(i,\cdot)} }  \abs{u_i^{((\gamma - 1.5)p\alpha d_2)}}^2 + 
	\sum_{j \in [d_2]}\sum_{i \in \Omega^*_{o(\cdot, j)} } \abs{v_j^{((\gamma - 1.5)p\alpha d_1)}}^2 \notag\\
	& \leq \sum_{i \in [d_1]} \frac{1.5}{\gamma - 1.5}\twonorm{u_i}^2 + \sum_{j \in [d_2]} \frac{1.5}{\gamma - 1.5}\twonorm{v_j}^2 \leq  \frac{3}{\gamma - 1.5} \frobnorm{\Proj_{\ObserveSet}(M - \MStar)}^2.
	\end{align}
	Moreover, we have
	\begin{align}
	& \frobnorm{\Proj_{\ObserveSet}(M - \MStar)}^2 \leq 2\frobnorm{\Proj_{\ObserveSet}(\UniqueU\DeltaV^{\top} + \DeltaU \UniqueV^{\top})}^2 + 2\frobnorm{\Proj_{\ObserveSet}(\DeltaU\DeltaV^{\top})}^2 \notag\\
	& \leq 2(1+\epsilon)p\frobnorm{\UniqueU\DeltaV^{\top} + \DeltaU \UniqueV^{\top}}^2 + 2p\frobnorm{\DeltaU}^2\frobnorm{\DeltaV}^2 + 18p\epsilon\sigma_1^* \frobnorm{\DeltaU}\frobnorm{\DeltaV} \notag\\
	& \leq 4(1+\epsilon)p\left(\opnorm{\UniqueU}^2 \frobnorm{\DeltaV}^2 + \opnorm{\UniqueV}^2 \frobnorm{\DeltaU}^2\right) + 2p\frobnorm{\DeltaU}^2\frobnorm{\DeltaV}^2 + 18p\epsilon\sigma_1^* \frobnorm{\DeltaU}\frobnorm{\DeltaV} \notag\\
	& \leq (4 + 13\epsilon)p\sigma_1^*\delta + 2p\delta^2, \label{eq:tmp7}
	\end{align}
	where the second step follows from Lemma \ref{lem:norm_sampling} and inequality \eqref{eq:uv_sampling} in Lemma \ref{lem:frobnorm_sampling}.
	Putting \eqref{eq:V2_partial}-\eqref{eq:tmp7} together, we obtain
	\begin{align*}
	T_2 & \leq (18\gamma + 14\beta + 27)\alpha \mu r \sigma_1^* \delta + \frac{3[(2 + 7\epsilon)\sigma_1^*\delta + \delta^2]}{\beta(\gamma - 1.5)}.
	\end{align*}
	
	\paragraph{Upper bound of $T_3$.} 
	By Cauchy-Schwarz inequality, we have
	\begin{align*}
	pT_3 & \leq \frobnorm{ \Proj_{\ObserveSet}(M - \MStar +  S - \SStar)} \frobnorm{\Proj_{\ObserveSet}( \DeltaU\DeltaV^{\top})} \\
	& \leq \frobnorm{ \Proj_{\ObserveSet}(M - \MStar +  S - \SStar)}  \sqrt{p\frobnorm{\DeltaU}^2\frobnorm{\DeltaV}^2 + 9p\epsilon\sigma_1^* \frobnorm{\DeltaU}\frobnorm{\DeltaV}} \\
	& \leq \frobnorm{ \Proj_{\ObserveSet}(M - \MStar +  S - \SStar)}\sqrt{p\delta^2 + 5p\epsilon\sigma_1^*\delta}.
	\end{align*}
	where we use \eqref{eq:uv_sampling} in Lemma \ref{lem:frobnorm_sampling} in the second step.
	
	We observe that $\Proj_{\ObserveSet}(M - \MStar + S - \SStar)$ is supported on $\ObserveSet\setminus\Omega$. Therefore, we have
	\begin{align*}
	\frobnorm{\Proj_{\ObserveSet}(M - \MStar +  S - \SStar)} & \leq \frobnorm{\Proj_{\ObserveSet \cap \Omega^c \cap \ObserveSet^{*c}}(M - \MStar)} + \frobnorm{ \Proj_{\ObserveSet \cap \Omega^c \cap \ObserveSet^{*}}(M - \MStar - \SStar)} \\
	& \leq \frobnorm{\Proj_{\ObserveSet}(M - \MStar)}+\frobnorm{ \Proj_{\Omega^c \cap \ObserveSet^{*}}(M - \MStar - \SStar)} \\
	& \leq \frobnorm{\Proj_{\ObserveSet}(M - \MStar)} + \sqrt{\frac{3}{\gamma -1.5}}\frobnorm{\Proj_{\ObserveSet}(M - \MStar)} \\
	& \leq \left(1 + \sqrt{\frac{3}{\gamma -1.5}}\right)\sqrt{(4 + 13\epsilon)p\sigma_1^*\delta + 2p\delta^2},
	\end{align*}
	where the third step follows from \eqref{eq:tmp6}, and the last step is from \eqref{eq:tmp7}. Under assumptions $\gamma = 3$, $\epsilon \leq 1/4$ and $\delta \leq \sigma_1^*$, we have
	\[
	T_3 \leq 3\sqrt{9\sigma_1^*\delta + 2\delta^2}\sqrt{\delta^2 + 5\epsilon\sigma_1^*\delta} \leq 10\sqrt{\sigma_1^*\delta^3} + 23\sqrt{\epsilon}\sigma_1^*\delta.
	\]
	
	\paragraph{Combining pieces.} Under the aforementioned assumptions, putting all pieces together leads to
	\begin{align*}
	& \trinprod{\nabla_{M} \mathcal{\widetilde{L}}(\PIU, \PIV; S)}{ \PIU \PIV^{\top} - \UniqueU\UniqueV^{\top} + \DeltaU \DeltaV^{\top} } \\
	& \geq \frac{3}{16}\frobnorm{M - \MStar}^2 - (14\beta + 81)\alpha \mu r \sigma_1^*\delta -\left(26\sqrt{\epsilon} + \frac{18}{\beta}\right)\sigma_1^*\delta - 10\sqrt{\sigma_1^*\delta^3} - 2\delta^2.
	\end{align*}	
\end{proof}

\subsection{Proof of Lemma \ref{lem:smoothness_partial}} \label{proof:lem:smoothness_partial}
Let $M := \PIU\PIV^{\top}$. We find that 
\begin{align*}
& \nabla_{\PIU}\widetilde{\Loss}(\PIU,\PIV;\PIS) = p^{-1}\Proj_{\ObserveSet}\left(M + S -  \MStar - \SStar\right)\PIV, \\
& \nabla_{\PIV}\widetilde{\Loss}(\PIU,\PIV;\PIS) = p^{-1}\Proj_{\ObserveSet}\left(M + S -  \MStar - \SStar\right)^{\top}\PIU.
\end{align*}
Conditioning on the event in Lemma \ref{lem:frobnorm_sampling}, since $(\PIU, \PIV) \in \widebar{\USet} \times \widebar{\VSet}$, inequalities \eqref{eq:Zv_sampling} and \eqref{eq:uZ_sampling} imply that
\[
\frobnorm{\nabla_{\PIU}\widetilde{\Loss}(\PIU,\PIV;\PIS)}^2 + \frobnorm{\nabla_{\PIV}\widetilde{\Loss}(\PIU,\PIV;\PIS)}^2 \leq \frac{12}{p}\mu r\sigma_1^*\frobnorm{\Proj_{\ObserveSet}\left(M + S -  \MStar - \SStar\right)}^2.
\]
It remains to bound the term  $\frobnorm{\Proj_{\ObserveSet}\left(M + S -  \MStar - \SStar\right)}^2$. Let $\ObserveSparse$ and $\Omega$ be the support of $\Proj_{\ObserveSet}(\SStar)$ and $S$ respectively. We observe that
\begin{align*}
\frobnorm{\Proj_{\ObserveSet}\left(M + S -  \MStar - \SStar\right)}^2 & = \frobnorm{\Proj_{\ObserveSparse \setminus \Omega}\left(M -  \MStar - \SStar\right)}^2 + \frobnorm{\Proj_{\ObserveSet^{*c} \cap \Omega^c \cap \ObserveSet}\left(M -  \MStar\right)}^2 \\
& \leq \frobnorm{\Proj_{\ObserveSparse \setminus \Omega}\left(M -  \MStar - \SStar\right)}^2 + \frobnorm{\Proj_{\ObserveSet}\left(M -  \MStar\right)}^2.
\end{align*}
In the proof of Lemma \ref{lem:local_descent_L_partial}, it is shown in \eqref{eq:tmp6} that
\[
\frobnorm{\Proj_{\ObserveSparse \setminus \Omega}\left(M -  \MStar - \SStar\right)}^2 \leq \frac{3}{\gamma - 1.5}\frobnorm{\Proj_{\ObserveSet}(M - \MStar)}^2.
\]
Moreover, following \eqref{eq:tmp7}, we have that
\begin{align*}
\frobnorm{\Proj_{\ObserveSet}(M - \MStar)}^2 & \leq 2(1+\epsilon)p\frobnorm{\UniqueU\DeltaV^{\top} + \DeltaU \UniqueV^{\top}}^2 + 2p\frobnorm{\DeltaU}^2\frobnorm{\DeltaV}^2 + 18p\epsilon\sigma_1^* \frobnorm{\DeltaU}\frobnorm{\DeltaV} \\
& \leq 4(1+\epsilon)p\frobnorm{M - \MStar}^2 + (6+4\epsilon)p\frobnorm{\DeltaU}^2\frobnorm{\DeltaV}^2 + 18p\epsilon\sigma_1^* \frobnorm{\DeltaU}\frobnorm{\DeltaV} \\
& \leq 4(1+\epsilon)p\frobnorm{M - \MStar}^2 + (6+4\epsilon)p\delta^2 + 9p\epsilon\sigma_1^*\delta.
\end{align*}
We thus finish proving our conclusion by combining all pieces and noticing that $\gamma = 3$ and $\epsilon \leq 1/4$.

\section{Proofs for Technical Lemmas} \label{sec:proof_lemmas}

In this section, we prove several technical lemmas that are used in the proofs of our main theorems.

\subsection{Proof of Lemma \ref{lem:op_to_inf}} \label{proof:lem:op_to_inf}
We observe that 
\[
\opnorm{A} = \sup_{x \in \mathbb{S}^{d_1 - 1}} \sup_{y \in \mathbb{S}^{d_2 - 1}} x^{\top}Ay.
\]
We denote the support of $A$ by $\Omega$. For any $x \in \real^{d_1}$, $y \in \real^{d_2}$ and $\beta > 0$, we have
\begin{align*}
x^{\top}Ay & = \sum_{(i,j) \in \Omega} x_i A_{(i,j)}y_j \leq \sum_{(i,j) \in \Omega}  \frac{1}{2}\infnorm{A}(\beta^{-1}x_i^2 + \beta y_j^2) \\
& = \frac{1}{2}\infnorm{A}\left(\sum_{i} \sum_{j \in \row{\Omega}{i}} \beta^{-1}x_i^2 + \sum_{j} \sum_{i \in \column{\Omega}{j}} \beta y_j^2 \right) \\
& \leq \frac{1}{2}\infnorm{A} \left(\alpha d_2 \beta^{-1}\twonorm{x}^2 + \alpha  d_1 \beta\twonorm{y}^2\right).
\end{align*}
It is thus implied that $\opnorm{A} \leq \frac{1}{2}\alpha(\beta^{-1}d_2 + \beta d_1)\infnorm{A}$. Choosing $\beta = \sqrt{d_2/d_1}$ completes the proof.

\subsection{Proof of Lemma \ref{lem:norm_sampling}} \label{proof:lem:norm_sampling}
We define a subspace $\mathcal{K} \subseteq \real^{d_1 \times d_2}$ as
\[
\mathcal{K} := \left\{X ~\big|~ X = \PIL^*A^{\top} + B\PIR^{*\top} ~\text{for some}~ A \in \real^{d_2 \times r}, B \in \real^{d_1 \times r}~\right\}. 
\]
Let $\Proj_{\mathcal{K}}$ be Euclidean projection onto $\mathcal{K}$. Then according to Theorem 4.1 in \cite{candes2009exact}, under our assumptions, for all matrices $X \in \real^{d_1 \times d_2}$, inequality
\begin{equation} \label{eq:exact_mc}
p^{-1}\frobnorm{\left(\Proj_{\mathcal{K}}\Proj_{\ObserveSet}\Proj_{\mathcal{K}} - p\Proj_{\mathcal{K}}\right)X} \leq \epsilon\frobnorm{X}
\end{equation}
holds with probability at least $1 - 2d^{-3}$. 

In our setting, by restricting $X = \PIL^*A^{\top} + B\PIR^{*\top}$, we have $\Proj_{\mathcal{K}}X = X$. Therefore, \eqref{eq:exact_mc} implies that
\[
\frobnorm{\Proj_{\mathcal{K}}\Proj_{\ObserveSet}X - pX} \leq p\epsilon\frobnorm{X}.
\]
For $\frobnorm{\Proj_{\ObserveSet}X}^2$, we have
\begin{align*}
\frobnorm{\Proj_{\ObserveSet}X}^2 & = \trinprod{\Proj_{\ObserveSet}X}{\Proj_{\ObserveSet}X} = \trinprod{\Proj_{\ObserveSet}X}{X} \\
& = \trinprod{\Proj_{\mathcal{K}}\Proj_{\ObserveSet}X}{X} \leq \frobnorm{\Proj_{\mathcal{K}}\Proj_{\ObserveSet}X}\frobnorm{X} \leq p(1+\epsilon)\frobnorm{X}^2.
\end{align*}
On the other hand, we have
\begin{align*}
\frobnorm{\Proj_{\ObserveSet}X}^2 & = \trinprod{\Proj_{\mathcal{K}}\Proj_{\ObserveSet}X}{X} = \trinprod{\Proj_{\mathcal{K}}\Proj_{\ObserveSet}X - pX + pX}{X} \\
& = p\frobnorm{X}^2 - \trinprod{X}{-\Proj_{\mathcal{K}}\Proj_{\ObserveSet}X + pX} \\
& \geq p\frobnorm{X}^2 - \frobnorm{X}\frobnorm{\Proj_{\mathcal{K}}\Proj_{\ObserveSet}X - pX} \geq p(1-\epsilon)\frobnorm{X}^2.
\end{align*}
Combining the above two inequalities, we complete the proof.

\subsection{Proof of Lemma \ref{lem:size_concentration}} \label{proof:lem:size_concentration}
We observe that $\abs{\row{\ObserveSet}{i}}$ is a summation of $d_2$ i.i.d. binary random variables with mean $p$ and variance $p(1-p)$. By Bernstein's inequality, for any $i \in [d_1]$,
\[
\Pr\left[ \big| \abs{\row{\ObserveSet}{i}} - pd_2\big| \geq \frac{1}{2}pd_2 \right] \leq 2\exp\left(-\frac{-\frac{1}{2}(pd_2/2)^2}{d_2p(1-p) + \frac{1}{3}(pd_2/2)}\right) \leq 2\exp\left(-\frac{3}{28}pd_2\right).
\]
By probabilistic union bound, we have
\[
\Pr\left[ \sup_{i \in [d_1]} \big|\abs{\row{\ObserveSet}{i}} - pd_2\big| \geq \frac{1}{2}pd_2\right] \leq 2d_1\exp\left(-\frac{3}{28}pd_2\right) \leq 2d^{-1},
\]
where the last inequality holds by assuming $p \geq \frac{56}{3}\frac{\log d}{d_2}$.

The term $\abs{\Omega^*_{o(i,\cdot)}}$ is a summation of at most $\alpha d_2$ i.i.d. binary random variables with mean $p$ and variance $p(1-p)$. Again, applying Bernstein's inequality leads to
\[
\Pr\left[ \abs{\Omega^*_{o(i,\cdot)}} - \mathbb{E}\left[\abs{\Omega^*_{o(i,\cdot)}}\right]  \geq \frac{1}{2}p\alpha d_2\right] \leq \exp\left(-\frac{3}{28}p\alpha d_2\right).
\]
Accordingly, by the assumption $p \geq \frac{56}{3}\frac{\log d}{\alpha d_2}$, we obtain
\[
\Pr\left[ \sup_{i \in [d_1]}\abs{\Omega^*_{o(i,\cdot)}} - pk  \geq \frac{1}{2}pk\right]  \leq d_1\exp\left(-\frac{3}{28}p\alpha d_2\right) \leq d^{-1}.
\]
The proofs for $\abs{\column{\ObserveSet}{j}}$ and $\abs{\Omega^*_{o(\cdot,j)}}$ follow the same idea.
\subsection{Proof of Lemma \ref{lem:frobnorm_sampling}} \label{proof:lem:frobnorm_sampling}
According to Lemma 3.2 in \cite{candes2011robust}, under condition $p \geq c_1\frac{\mu \log d}{d_1 \wedge d_2}$, for any fixed matrix $A \in \real^{d_1 \times d_2}$, we have
\[
\opnorm{A - p^{-1}\Proj_{\ObserveSet}A} \leq c_2\sqrt{\frac{d\log d}{p}}\infnorm{A},
\]
holds with probability at least $1 - \order(d^{-3})$. Letting $A$ be all-ones matrix, then we have that for all $u \in \real^{d_1}, v \in \real^{d_2}$,
\[
\sum_{(i,j)\in \ObserveSet} u_i v_j \leq p\onenorm{u}\onenorm{v} + c_2\sqrt{pd\log d}\twonorm{u}\twonorm{v}.
\]
We find that
\begin{align*}
\frobnorm{\Proj_{\ObserveSet}(\PIU \PIV^{\top})}^2 & \leq \sum_{(i,j)\in \ObserveSet}\twonorm{\row{\PIU}{i}}^2\twonorm{\row{\PIV}{j}}^2 \\
& \leq p \frobnorm{\PIU}^2\frobnorm{\PIV}^2 + c_2\sqrt{pd\log d} \sqrt{\sum_{i \in [d_1]} \twonorm{\row{\PIU}{i}}^4}\sqrt{\sum_{j \in [d_2]} \twonorm{\row{\PIV}{j}}^4} \\
& \leq p \frobnorm{\PIU}^2\frobnorm{\PIV}^2 + c_2\sqrt{pd \log d}\frobnorm{\PIU}\frobnorm{\PIV}\twoinfnorm{\PIU}\twoinfnorm{\PIV} \\
& \leq p \frobnorm{\PIU}^2\frobnorm{\PIV}^2 + c_2\sqrt{\frac{p \mu^2r^2d\log d}{d_1 d_2}}\frobnorm{\PIU}\frobnorm{\PIV}.
\end{align*}
By the assumption $p \gtrsim \frac{\mu^2r^2\log d}{\epsilon^2(d_1 \wedge d_2)}$, we finish proving \eqref{eq:uv_sampling}.

According to the proof of Lemma \ref{lem:size_concentration}, if $p \geq c\frac{\log d}{d_1\wedge d_2}$, with probability at least $1 - \order(d^{-1})$, we have $\abs{\row{\ObserveSet}{i}} \leq \frac{3}{2}pd_2$ and $\abs{\column{\ObserveSet}{j}} \leq \frac{3}{2}pd_1$ for all $i \in [d_1]$ and $j \in [d_2]$. Conditioning on this event, we have
\begin{align*}
\frobnorm{\Proj_{\ObserveSet}(Z)V}^2 & = \sum_{i \in [d_1]} \sum_{k \in [r]} \inprod{\row{(\Proj_{\ObserveSet}(Z))}{i}}{\column{H}{k}}^2 \\
& \leq \sum_{i \in [d_1]} \sum_{k \in [r]} \twonorm{\row{(\Proj_{\ObserveSet}(Z))}{i}}^2 \sum_{j \in \row{\Omega}{i}} \entry{\PIV}{j}{k}^2 \\
& = \frobnorm{\Proj_{\ObserveSet}Z}^2 \sum_{j \in \row{\Omega}{i}} \twonorm{\row{\PIV}{i}}^2 \\
& \leq  \frobnorm{\Proj_{\ObserveSet}Z}^2 \frac{3}{2}pd_2 \cdot \twoinfnorm{\PIV}^2 \leq 2\mu r p \frobnorm{\Proj_{\ObserveSet}Z}^2.
\end{align*}
We thus finish proving \eqref{eq:Zv_sampling}. Inequality \eqref{eq:uZ_sampling} can be proved in the same way.

\subsection{Proof of Lemma \ref{lem:symm}} \label{proof:lem:symm}
Recall that we let $F := [\PIU; \PIV]$ and $F_{\pi^*} := [\OptU; \OptV]\RotateM$ for some matrix $\RotateM \in \RotateSet{r}$, which minimizes the following function
\begin{equation} \label{eq:opt_Q}
\frobnorm{F - [\OptU; \OptV]Q}^2.
\end{equation}
Let $F^* := [\OptU; \OptV]$. Expanding the above term, we find that $\RotateM$ is the maximizer of $\trinprod{F}{F^*Q} = \Trace(F^{\top}F^{*}Q)$. Suppose $F^{\top}F^{*}$ has SVD with form $Q_1\Lambda Q_2^{\top}$ for $Q_1, Q_2 \in \RotateSet{r}$. When the minimum diagonal term of $\Lambda$ is positive, we conclude that the minimizer of \eqref{eq:opt_Q} is unique and $Q = Q_2Q_1^{\top}$. To prove this argument, we note that 
\[
\Trace(F^{\top}F^{*}Q) = \sum_{i \in [r]} \entry{\Lambda}{i}{i}\inprod{p_i}{q_i},
\]
where $p_i$ is the $i$-th column of $Q_1$ and $q_i$ is the $i$-th column of $Q^{\top}Q_2$. Hence, $\Trace(F^{\top}F^{*}Q) \leq \sum_{i \in [r]} \entry{\Lambda}{i}{i}$ and the equality holds if and only if $p_i = q_i$ for all $i \in [r]$ since every $\entry{\Lambda}{i}{i} > 0$. We have $Q_1 = Q^{\top}Q_2$ and thus finish proving the argument.

Under our assumption $\opnorm{F - F_{\pi^*}} < \sqrt{2\sigma_r^*}$, for any nonzero vector $u \in \real^r$, we have
\[
\twonorm{F^{\top}F_{\pi^*}u} \geq \twonorm{F_{\pi^*}^{\top}F_{\pi^*}u} - \twonorm{(F_{\pi^*} - F)^{\top}F_{\pi^*}u} \geq (\sqrt{2\sigma_r^*} - \opnorm{F_{\pi^*} - F})\frobnorm{F_{\pi^*}u} > 0.
\]
In the second step, we use the fact that the singular values of $F_{\pi^*}$ are equal to the diagonal terms of $\sqrt{2}\SVMat^{*1/2}$. Hence, $F^{\top}F_{\pi^*}$ has full rank. Furthermore, it implies that $F^{\top}F^*$ has full rank and only contains positive singular values.

Proceeding with the proved argument, we have
\[
F^{\top}F_{\pi^*} = Q_1\Lambda Q_2^{\top}Q_2Q_1^{\top} = Q_1\Lambda Q_1^{\top},
\]
which implies that $F^{\top}F_{\pi^*}$ is symmetric. Accordingly, we have $(F - F_{\pi^*})^{\top}F_{\pi^*}$ is also symmetric.

\section*{Acknowledgment}

Y. Chen acknowledges support from the School of Operations Research and Information Engineering, Cornell University.

%%%%%%%%%%%%%%%%%%%%%%%%%%%%%%%%%%%%%%%%%%%%%%%%%%%%%%%%%%%%%%%%%%%%%%%%%%%

% FOR BIBTEX
%\bibliographystyle{plain}
%\bibliography{rpca}

% FOR BIBLATEX
\printbibliography

%%%%%%%%%%%%%%%%%%%%%%%%%%%%%%%%%%%%%%%%%%%%%%%%%%%%%%%%%%%%%%%%%%%%%%%%%%%

\newpage

\appendixpage

\appendix
\section{Supporting Lemmas} \label{sec:supporting_lemma}
In this section, we provide several technical lemmas used for proving our main results.
\begin{lemma} \label{lem:delta_norm}
	For any $(\OptU, \OptV) \in \OptSet(\MStar)$, $\PIU \in \real^{d_1 \times r}$ and $\PIV \in \real^{d_2 \times r}$, we have
	\[
	\frobnorm{\PIU\PIV^{\top} - \OptU \OptVT} \leq \sqrt{\sigma_1^*}(\frobnorm{\DeltaV} + \frobnorm{\DeltaU}) + \frobnorm{\DeltaU}\frobnorm{\DeltaV},
	\]
	where $\DeltaU := \PIU - \OptU$, $\DeltaV := \PIV - \OptV$.
\end{lemma}
\begin{proof}
	We observe that $\PIU\PIV^{\top} - \OptU\OptVT = \OptU \DeltaV^{\top} + \DeltaU \OptVT + \DeltaU\DeltaV^{\top}$. Hence,
	\begin{align*}
	& \frobnorm{\PIU\PIV^{\top} - \OptU\OptVT} \leq \frobnorm{\OptU \DeltaV^{\top}}  + \frobnorm{\DeltaU \OptVT} + \frobnorm{\DeltaU\DeltaV^{\top}} \\
	& \leq \opnorm{\OptU}\frobnorm{\DeltaV} + \opnorm{\OptV}\frobnorm{\DeltaU} + \frobnorm{\DeltaU}\frobnorm{\DeltaV}.
	\end{align*}
\end{proof}
Furthermore, assuming $(\PIU, \PIV) \in \USet \times \VSet$, where $\USet$ and $\VSet$ satisfy the conditions in \eqref{eq:uvcondition}, we have the next result.
\begin{lemma} \label{lem:infity_norm_bound}
	For any $(i, j) \in [d_1]\times [d_2]$, we have
	\begin{equation} \label{eq:infity_bound}
	\abs{\entry{(\PIU \PIV^{\top} - \OptU \OptVT)}{i}{j}}\leq 3\sqrt{\frac{\mu r\sigma_1^*}{d_1}} \twonorm{\Delta_{V(j,\cdot)}} + 3\sqrt{\frac{\mu r\sigma_1^*}{d_2}} \twonorm{\Delta_{U(i,\cdot)}}
	\end{equation}
\end{lemma}
\begin{proof}
	We observe that
	\begin{align*}
	& \abs{\entry{(\PIU \PIV^{\top} - \MStar)}{i}{j}} \leq \abs{\inprod{\row{\OptU}{i}}{\Delta_{V(j,\cdot)}}} + \abs{\inprod{\row{\OptV}{j}}{\Delta_{U(i,\cdot)}}} +\abs{\inprod{\Delta_{U(i,\cdot)}}{\Delta_{V(j,\cdot)}}} \\
	& \leq \sqrt{\frac{\mu r\sigma_1^*}{d_1}} \twonorm{\Delta_{V(j,\cdot)}} + \sqrt{\frac{\mu r\sigma_1^*}{d_2}} \twonorm{\Delta_{U(i,\cdot)}} + \frac{1}{2}\twoinfnorm{\DeltaU}\twonorm{\Delta_{V(j,\cdot)}} + \frac{1}{2}\twoinfnorm{\DeltaV}\twonorm{\Delta_{U(i,\cdot)}}.
	\end{align*}
	By noticing that 
	\[
	\twoinfnorm{\DeltaU} \leq \twoinfnorm{\OptU} + \twoinfnorm{\PIU} \leq 3\sqrt{\frac{\mu r \sigma_1^*}{d_1}}, ~~ \twoinfnorm{\DeltaV} \leq \twoinfnorm{\OptV} + \twoinfnorm{\PIV} \leq 3\sqrt{\frac{\mu r \sigma_1^*}{d_2}},
	\]
	we complete the proof.
\end{proof}

Lemma \ref{lem:infity_norm_bound} can be used to prove the following result.
\begin{lemma} \label{lem:projection}
	For any $\alpha  \in [0,1]$, suppose $\Omega \subseteq [d_1] \times [d_2]$ satisfies $\abs{\row{\Omega}{i}} \leq \alpha d_2$ for all $i \in [d_1]$ and $\abs{\column{\Omega}{j}} \leq \alpha d_1$ for all $j \in [d_2]$. Then we have
	\[
	\frobnorm{\Proj_{\Omega}(\PIU\PIV^{\top} - \OptU\OptVT)}^2 \leq 18\alpha \mu r  \sigma_1^*(\frobnorm{\DeltaV}^2 + \frobnorm{\DeltaU}^2).
	\]
\end{lemma}
\begin{proof}
	Using Lemma \ref{lem:infity_norm_bound} for bounding each entry of $\PIU\PIV^{\top} - \OptU\OptVT$, we have that
	\begin{align*}
	& \frobnorm{\Proj_{\Omega}(\PIU\PIV^{\top} - \OptU\OptVT)}^2 \leq \sum_{(i,j) \in \Omega} \abs{\entry{(\PIU\PIV^{\top} - \OptU\OptVT)}{i}{j}}^2 \\
	& \leq \sum_{(i,j) \in \Omega} \frac{18\mu r \sigma_1^*}{d_1}\twonorm{\Delta_{V(j,\cdot)}}^2 + \frac{18\mu r \sigma_1^*}{d_2}\twonorm{\Delta_{U(i,\cdot)}}^2\\
	& \leq \sum_{j} \sum_{i \in \column{\Omega}{j}} \frac{18\mu r \sigma_1^*}{d_1}\twonorm{\Delta_{V(j,\cdot)}}^2 + \sum_{i} \sum_{j \in \row{\Omega}{i}} \frac{18\mu r \sigma_1^*}{d_2}\twonorm{\Delta_{U(i,\cdot)}}^2 \\
	& \leq 18\alpha \mu r  \sigma_1^*(\frobnorm{\DeltaV}^2 + \frobnorm{\DeltaU}^2).
	\end{align*}
\end{proof}

Denote the $i$-th largest singular value of matrix $M$ by $\sigma_i(M)$.
\begin{lemma}[Lemma 5.14 in \cite{tu2015low}] \label{lem:tu_lemma}
	Let $M_1, M_2 \in \real^{d_1 \times d_2}$ be two rank $r$ matrices. Suppose they have SVDs $M_1 = \PIL_1\Sigma_1\PIR_1^{\top}$ and $M_2 = \PIL_2\Sigma_2\PIR_2^{\top}$. Suppose $\opnorm{M_1 - M_2} \leq \frac{1}{2}\sigma_r(M_1)$. Then we have
	\[
	d^2(\PIL_2\Sigma_2^{1/2}, \PIR_2\Sigma_2^{1/2}; \PIL_1\Sigma_1^{1/2}, \PIR_1\Sigma_1^{1/2}) \leq \frac{2}{\sqrt{2} - 1}\frac{\frobnorm{M_2 - M_1}^2}{\sigma_r(M_1)}.
	\]
\end{lemma}

\section{Parameter Settings for FB Separation Experiments} \label{sec:fb_more}

We approximate the FB separation problem by the RPCA framework with $r = 10$, $\alpha = 0.2$, $\mu = 10$. Our algorithmic parameters are set as $\gamma = 1$, $\eta = 1/(2 \hat{\sigma}_1^*)$, where $\hat{\sigma}_1^*$ is an estimate of $\sigma_1^*$ obtained from the initial SVD. The parameters of AltProj are kept as provided in the default setting. For IALM, we use the tradeoff paramter $\lambda = 1/\sqrt{d_1}$, where $d_1$ is the number of pixels in each frame (the number of rows in $Y$).

Note that both IALM and AltProj use the stopping criterion 
$$\frobnorm{Y-M_{t}-S_{t}} / \frobnorm{Y} \le 10^{-3}.$$
Our algorithm for the partial observation setting never explicitly forms the $ d_1 $-by-$d_2$ matrix $ M_{t} =\PIU_{t}\PIV_{t} ^\top$, which is favored in large scale problems, but also renders the above criterion inapplicable. Instead, we use the following stopping criterion
\[
\frac{\frobnorm{\PIU_{t+1} - \PIU_t}^2 + \frobnorm{\PIV_{t+1} - \PIV_t}^2}{\frobnorm{\PIU_t}^2 + \frobnorm{\PIV_t}^2} \le 4 \times 10^{-4}.
\]
This rule checks whether the iterates corresponding to low-rank factors become stable.  In fact, our stopping criterion seems more natural and practical because in most real applications, matrix $Y$ cannot be strictly decomposed into low-rank $M$ and sparse $S$ that satisfy $Y = M+S$. Instead of forcing $M + S$ to be close to $Y$, our rule relies on seeking a robust subspace that captures the most variance of $Y$.

\end{document}